\documentclass[12pt]{article}

\usepackage[ruled,vlined]{algorithm2e}
\usepackage{float}
\usepackage{graphicx}
\usepackage{verbatim}
\usepackage[sumlimits,]{amsmath}
\usepackage{amsfonts}
\usepackage[numbers]{natbib}
\usepackage{lineno}
\RequirePackage[OT1]{fontenc}
\RequirePackage{amsthm,amsmath,amssymb}
\RequirePackage[colorlinks,citecolor=blue,urlcolor=blue]{hyperref}

\def \d { {\,\mbox{d}}}

\def \dT { \,\mbox{d}T}

\def \dx { \, \mbox{d}x}

\def \supp { {\mbox{supp}}}

\def \calA {\mathcal{A}}

\def \calC {\mathcal{C}}

\def \calE {\mathcal{E}}

\def \calI { {\mathcal I}}

\def \calN { {\mathcal N}}

\def \calS { {\mathcal S}}
\def \calT {\mathcal{T}}

\def \calV { \mathcal{V}}

\def \eps { {\varepsilon}}

\def \st {:\,}

\def \iid { {i.i.d.~}}

\def \rmP {\mathrm{P}}

\def \E { {\mathbb E}}
\def \g { {\,|\,}}
\def \p {\partial}
\def \equalindistribution { {\overset{d}{=}}}

\newcommand{\Exp}[1]{ \E\left\{ #1 \right\} }

\newcommand{\Trace}[1] { {\mbox{Trace}\left\{ #1 \right\}}}
\newcommand{\Var}[1]{ \mbox{Var}\left\{ #1 \right\} }

\def \Rone { {\mathbb R}}

\def \RK { {{\mathbb R}^K}}

\def \RN {{\mathbb R}^N}
\def \RM {{\mathbb R}^M}

\newcommand{\pushfwd}[1]{ {#1_{\#}}}

\newcommand{\Wishart}[2]{ Wishart(#1, #2) }
\newcommand{\InverseWishart}[2]{ InverseWishart(#1, #2) }
\def \accept { {\mathrm{Accept}}}
\def \Chat { {\hat{C}}}
\def \Dhat { {\hat{D}}}
\def \ess { {\mathrm{ESS}} }
\def \miness { {\min_n\left\{ \mathrm{ESS_n} \right\}}}
\def \Lhat { {\hat{L}}}
\def \Rhat { {\hat{R}}}

\def \swap { {\mathrm{Swap}}}

\def \W  { {\kappa}}

\def \prior { {p(x)}}
\def \logprior { {\log p(x)}}
\def \posterior { {p(x\g y)}}
\def \swapprob { {\rmP[\swap_{(r, r+1)}]}}
\def \logposterior { {\log p(x\g y)}}
\def \likelihood { {p(y\g x)}}

\def \RMN { {\Rone^{M\times N}}}
\def \RNN { {\Rone^{N\times N}}}
\def \Tmax { {T_{max}}}

\theoremstyle{plain}
\newtheorem{theorem}{Theorem}[section]
\newtheorem{lemma}{Lemma}[section]

\newtheorem{proposition}{Proposition}[section]
\theoremstyle{remark}

\providecommand{\keywords}[1]
{
  \small	
  \textbf{\textit{Keywords---}} #1
}

\begin{document}

\title{Hamiltonian Monte Carlo in Inverse Problems; Ill-Conditioning and Multi-Modality.}
\author{
Langmore, I.~\footnote{\texttt{ianlangmore@gmail.com}}
\quad Dikovsky, M.
\quad Geraedts, S
\quad Norgaard, P
\\
\quad von Behren, R.
\\
Google Research
}

\maketitle
\date

\begin{abstract}
  The Hamiltonian Monte Carlo (HMC) method allows sampling from continuous densities. Favorable scaling with dimension has led to wide adoption of HMC by the statistics community. Modern auto-differentiating software should allow more widespread usage in Bayesian inverse problems. This paper analyzes two major difficulties encountered using HMC for inverse problems: poor conditioning and multi-modality. Novel results on preconditioning and replica exchange Monte Carlo parameter selection are presented in the context of spectroscopy. Recommendations are given for the number of integration steps as well as step size, preconditioner type and fitting, annealing form and schedule. These recommendations are analyzed rigorously in the Gaussian case, and shown to generalize in a fusion plasma reconstruction.
\end{abstract}



\section{Introduction}
\label{section:introduction}

The goal of Bayesian inverse problems is to produce and characterize $\posterior$, the \emph{posterior} distribution over possible state variables $X$, given measurements $y$. In principle, samples from the posterior can be used to determine the mean, quantiles, and other relevant statistics. These samples can be obtained using Markov Chain Monte Carlo sampling, which requires only that the log density, $\logposterior$ (assumed absolutely continuous with respect to Lebesgue measure), be available as a function (up to an additive constant). In practice, extracting samples can take prohibitively long, so people often resort to point estimates.

The most common Monte Carlo setup is \emph{random walk Metropolis-Hastings}. This requires, once burnt-in, $O(N)$ evaluations of $\logposterior$ for each effective sample $X\in \RN$. A more favorable scaling is obtained using \emph{Hamiltonian Monte Carlo}, or HMC, which requires only $O(N^{1/4})$ evaluations of $\nabla_z\logposterior$ \cite{Neal2011-gq}.
Due to the $O(N^{1/4})$ scaling, HMC has seen wide acceptance in the statistics community.
However, thus far, HMC has seen only minimal usage in the world of inverse problems \cite{Fichtner2019-ef,Bui-Thanh2014-vk,Nagel_Joseph_B2016-bu}. The barriers to acceptance are real: The gradient evaluations, $\nabla_z\logposterior$, required by HMC, translate to derivatives through a forward model.  This task is made easier by recent advancements in auto-differentiating software. The next two barriers are geometrical, and appear in somewhat predictable ways in inverse problems.  In particular, the second barrier is ill-conditioned posterior covariance, often induced by low rank and/or low noise forward models. These introduce a large multiplier to the $O(N^{1/4})$ scaling that must be dealt with.
Another barrier is that of posterior mass separated by regions of extremely low density. When present, this \emph{multi-modality} causes such difficulty that the highest priority is crossing these low density regions.

Our efforts to mitigate ill-conditioned covariance led to reparameterization and linear preconditioning. We contribute rigorous analysis of reparameterization, diagonal, and ``full covariance'' preconditioning, including an algorithm to select burn-in size.
\emph{Replica Exchange Monte Carlo} (abbreviated REMC, also known as \emph{parallel tempering}), is used to deal with multi-modality. Novel criteria for selecting annealing form and schedule, number of integration steps as well as step size are given. In all cases, a generic Gaussian problem is used to extract concrete recommendations, which we then test on a spectroscopy-based inversion.

Our perspective is the result of Google's ongoing work in reconstruction of (fusion) plasma state \cite{Dikovsky2021-cm}. In this industrial research setting, new experimental data arrives daily. Reconstructions must be done, and re-done, for \emph{thousands} of experiments.  Unexpected artifacts may appear in new measurements that are not well represented by the current model.  Physicists modify models weekly and need to understand changes. Sampling code must work well in the majority of reconstructions, and optimal tuning in each situation is not possible.  For that reason, we emphasize the use of simple scaling laws and crude algorithmic decisions over intricate methods. 
With this perspective, we add to the body of inverse problems oriented HMC literature.
See \cite{Fichtner2019-ef,Bui-Thanh2014-rg,Nagel_Joseph_B2016-bu,Conway2018-sk} for some applications.
See also \cite{Beskos2011-kg} for reparameterizations of HMC relevant to inverse problems in high dimension.
See \cite{Au2020-hc} for a lift-and-project approach applicable to inverse problems with low noise.
An introduction to HMC can be found in \cite{Neal2011-gq}. For those with some HMC experience we recommend \cite{Betancourt2017-rh}.
An introduction to Bayesian inverse problems can be found in \cite{Calvetti2018-gg}, and an overview of Bayesian modeling in statistics in \cite{Gelman2020-jm}.

Section \ref{section:hmc-step-size} is a prerequisite for sections \ref{section:sampling-with-correlations} and \ref{section:sampling-with-multi-modality}. All other sections can be read independently.

Section \ref{section:geometry-of-ip} describes how common features of inverse problems can lead to difficult sampling situations. Section \ref{section:hmc-step-size} reviews HMC. Section \ref{section:sampling-with-correlations} discusses preconditioning. Section \ref{section:sampling-with-multi-modality} discusses REMC.

\section{Some features found in inverse problems}
\label{section:geometry-of-ip}

In the Bayesian inverse problems setup we are interested in, the unknown $X$ parameterizes a physical quantity of interest. For each unique $X$, a \emph{single} measurement $Y\in\RM$ is observed. The likelihood is chosen as a (hopefully very accurate) representation of the data generating processes. Often, little to no ground truth examples of $X$ are available.  A different situation is often found in the statistics community, where many independent $Y$ are observed for the one and only $X$. For example, $X$ could be a coefficient of effectiveness of a drug treatment, and $Y$ the outcome (recovery or not). The likelihood is often then chosen as a mathematically convenient representation of the distribution of possible $Y$, given $X$  \cite{Gelman2020-jm}. Often, no attempt is made to describe details of the process by which $Y$ emerges. Both of these contrast with the so-called ``generative models'' popular in the machine learning community. There, the usual goal is to learn a probabilistic relationship $X\to Y$ so as to generate new $Y$ \cite{Goodfellow2016-lc}. The likelihood there is often taken as a mathematically convenient and extremely flexible function with \emph{learned} parameters.

These differences lead to a different prioritization of challenges for Bayesian inverse problems.
The lack of flexibility in choosing a likelihood means we must deal with the covariance structure imposed by the measurement. The need for physically realistic representations of $X$, combined with often sparse measurements, means we often use parameterizations leading to multiple local maxima of $\posterior$. Contrast this with the ``generative models'' world. There, a deep network can often transform a simple posterior to the observed $Y$, thus reducing the need to model correlations \cite{Farquhar2020-ex}. Moreover, since recovering $X$ is a non-goal, being stuck in one of many modes is okay, so long as the resultant distribution over $Y$ is the same.

\subsection{Poorly conditioned posterior covariance}
\label{subsection:covariance-structures}
Consider the toy problem
\begin{align}
  \label{align:gaussian-inverse-problem}
  \begin{split}
    \prior &\propto \exp\left\{ -\frac{1}{2}x^T C_{pr}^{-1}x \right\},
    \quad \likelihood \propto \exp\left\{ -\frac{1}{2\sigma^2}\|Ax - y\|^2 \right\}.
  \end{split}
\end{align}
where $A\in \RMN$ is the \emph{forward matrix}, mapping the unknown state variable $X$ to the measurement space $\RM$. The prior correlation, $C_{pr}$, is the discretization of \eqref{align:prior-covariance}.
Collecting terms quadratic in $x$, the posterior covariance is
\begin{align}
  C &= \left[ C_{pr}^{-1} + \sigma^{-2}A^TA \right]^{-1}.
  \label{align:gaussian-posterior-covariance}
\end{align}
Suppose the forward matrix $A$ is low rank, and so only constrains $M < N$ directions in $\RN$. For small $\sigma$, these constrained directions will have small posterior variance asymptotic to $\sigma^2$. The other directions will have variance governed by $C_{pr}$. The result is poorly conditioned posterior covariance.  See figure \ref{fig:gaussian-posterior-covariance-spectra}. In that example, sampling is about 92 times less efficient than if the posterior was ideally conditioned. Section \ref{subsection:hmc-step-size-scaling-laws} details the relationship between sampling efficiency and conditioning.
\begin{figure}[h]
  \begin{center}
    \includegraphics[width=\textwidth]{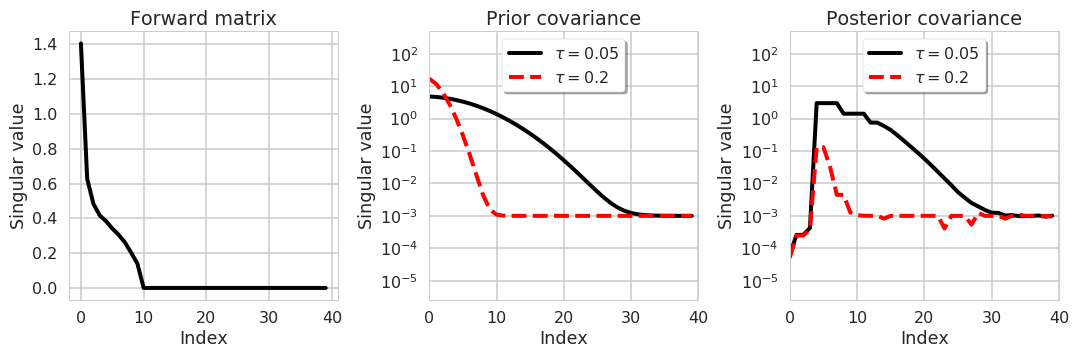}
  \end{center}
  \caption{
  {\bf Singular values of Gaussian problem:}
  Spectra of relevant matrices in \eqref{align:gaussian-inverse-problem}.
  The low rank forward matrix and/or long prior correlation length results in poor conditioning.
  In all cases, $M=20$, $N=40$, and the prior shrinkage $\delta=0.001$.
  We sort spectra by their corresponding eigenvectors' (Pearson) correlation with prior eigenvectors.  Frequency increases left to right.
  {\bf Left:} The forward matrix $A\in\RMN$.
  {\bf Center:} Prior covariance for two different values of prior correlation length $\tau$, c.f.~\eqref{align:prior-covariance}.
  {\bf Right:} Posterior covariance \eqref{align:gaussian-posterior-covariance} is poorly conditioned.
  }
  \label{fig:gaussian-posterior-covariance-spectra}
\end{figure}

\subsection{Multi-modal posteriors}
\label{subsection:multi-modal-posteriors}
The posterior corresponding to the linear problem from section \ref{subsection:covariance-structures} will always have exactly one local maximum, or \emph{mode}. In practice, non-linear parameterizations of state can easily lead to multiple modes. These tend to arise in an attempt to model details that are not fully constrained by the data.

A multi-modal toy model is a Gaussian prior for $X\in\RN$ and mixture of Gaussians likelihood corresponding to $M\leq N$ measurements.
\begin{align}
  \label{align:multi-model-toy-model}
  \begin{split}
    p(x) &\propto \prod_{n=1}^N e^{-x_n^2/2},
    \quad
    p(y\g x) \propto \prod_{m=1}^M \left[ e^{-(x_m-1)^2/(2\sigma^2)} + e^{-(x_m+1)^2/(2\sigma^2)} \right].
  \end{split}
\end{align}
This leads to posterior density
\begin{align*}
  p(x\g y) &\propto
  \left( \prod_{m=1}^M \left[ e^{(x_m - \mu(\sigma))^2 / (2 v(\sigma)^2)} +  e^{(x_m + \mu(\sigma))^2 / (2 v(\sigma)^2)} \right]
  \right)
  \left( \prod_{n=M+1}^N e^{-x_n^2/2}\right),\\
  \mu(\sigma) :&= 1 + \sigma^2,\qquad
  v(\sigma)^2 := \sigma^2 / (1 + \sigma^2).
\end{align*}
When $\sigma\ll1$, standard HMC will find itself stuck in any of the $2^M$ modes.  The probability of escaping one of these modes is, to leading order, less than $\exp\left\{ - (1 + \sigma^2)/\sigma^2\right\}$. See section \ref{subsection:conductance-basics}.

\subsection{Model problem: Reconstruction from spectroscopy}
\label{subsection:model-problem}
Our model problem is the inference of amplitude, temperature, and velocity of an ion species from an emission spectrum. Photons are emitted from an N-point discretization of the square $[-1,1]\times[-1,1]$, and line integrated emission is measured. Some details are given, so as to emphasize the ability of HMC to handle non-Gaussian posteriors. Other sections do \emph{not} require understanding of these details.

For a set of (around 200) frequencies $\nu$, we parameterize emissivity as a transformation, $\varphi_\nu$, of an \emph{a-priori} Gaussian random variable $X\in\RN$, $X\sim\calN(0, C_{pr})$. At each frequency $\nu$, an integration matrix $\calI\in\RMN$ projects the emissivity $\varphi_\nu(X)$ onto $M$ measurements. This models $M$ distinct view \emph{chords} in the $[-1,1]\times[-1,1]$ square (Fig.~\ref{fig:model-problem-basics}).
\begin{align*}
  Y_\nu = \calI\varphi_\nu(X) + \eta_\nu,
\end{align*}
The noise $\eta_\nu\in\RM$ is, conditional on $X$, Gaussian, and independent for every frequency and chord.  It takes the form
\begin{align*}
  \eta_{\nu, m} \sim \calN(0, \sigma^2 + \sigma_p^2 [\calI\varphi_\nu(X)]_m^2).
\end{align*}
The factor $\sigma$ gives rise to the usual additive independent noise. $\sigma_p$ gives us noise proportional to the signal, representing model error. This leads to the likelihood
\begin{align}
  \label{align:spectroscopic-likelihood}
  \begin{split}
    p(y\g x) &=
    \prod_{m=1}^M
    \frac{\exp\left\{ -\frac{1}{2\sigma(x)^2}(y_m - \left[\calI\varphi_\nu(x)\right]_m)^2 \right\}}{\sqrt{2\pi \sigma(x)_m^2 }},\\
    \sigma(x)_m^2 :&= \sigma^2 + \sigma_p^2\left[ \calI\varphi_\nu(x) \right]_m^2.
  \end{split}
\end{align}

The mapping $\varphi_\nu$ is composed of two stages. First, $X$ is divided into independent amplitude, temperature, and velocity components; $X=(X_\calA, X_\calT, X_\calV)$. These are mapped to amplitude $\calA$, temperature $\calT$, and velocity $\calV$.  The mapping is either \emph{slab}, meaning constant on chord-aligned rectangles, or \emph{shell}, meaning constant on circles of rotation, about a shifted center (Fig.~\ref{fig:model-problem-basics}).
\begin{figure}[H]
  \begin{center}
    \includegraphics[width=\textwidth]{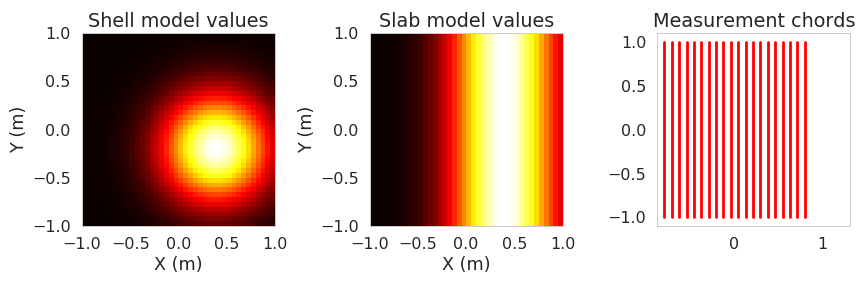}
  \end{center}
  \caption{
  {\bf Left:} Example shell model values for arbitrary plasma state variable (amplitude, temperature, or density). They are radially smooth and circularly symmetric, about a shifted center.
  {\bf Center:} Example slab model values are constant along rectangles aligned with the measurement chords.
  {\bf Right:} 20 rays through a disk shaped body, representing chords of measurement.
  }
  \label{fig:model-problem-basics}
\end{figure}

To make amplitude and temperature positive, we use the \emph{softplus} function; $\calS(u) := \log[1 + exp\left\{ u \right\}]$.
\begin{align*}
  \calA :&= c_\calA \,\calS(\tilde X_\calA), \quad
  \calT := c_\calT \,\calS(\tilde X_\calT), \quad
  \calV := c_\calV \,\tilde X_\calV.
\end{align*}
Above, $(c_\calA, c_\calT, c_\calV)$ are scaling factors, and $\tilde X := (\tilde X_\calA, \tilde X_\calT, \tilde X_\calV)$ is a transformation of $X$ to impart radial correlation and a shifted center. That is, the amplitude, temperature, and velocity components of $\tilde X$ have radial values approximating a zero mean Gaussian random field with correlation
\begin{align}
  \Gamma_{pr}(r, r') :&= \exp\left\{ -\frac{|r - r'|^2}{2\tau^2} \right\} + \delta.
  \label{align:prior-covariance}
\end{align}
Above, $\delta>0$ is a \emph{shrinkage} parameter preventing small eigenvalues from creating numerical difficulties. The coordinate $r$ is the radius about a center shifted by an \emph{a-priori} 2-D Normal.

The second step of $\varphi_\nu$ is the dimensionless discrete emissivity spectrum.
\begin{align}
  \label{align:emissivity-equation}
  \calE_\nu(\calA, \calT, \calV) :&= \frac{\calA}{\sqrt{2 \pi} w} \exp \left( \frac{-(\nu^{-1} - \tilde\nu^{-1})^2}{2 w^2} \right),
\end{align}
where $\tilde\nu = \nu_0 / (1 - \calV)$ is the Doppler shifted center frequency, and $w = \sqrt{\calT} / \nu_0$ is the Doppler broadened bandwidth.

In its shell parameterization, the spectroscopy problem can suffer from ill-conditioned posterior covariance.
The Gaussian toy problem \eqref{align:gaussian-inverse-problem} is a linearization, achieved by assuming amplitude $\calA$ is the only unknown, ignoring the softplus, and setting $\sigma_p=0$ and $c_\calA=1$.

Multi-modality can also occur. Suppose the measured emissivity $Y_\nu$ has two spectral peaks at $\nu_1$ and $\nu_2$, as a result of Carbon III and an additional, unexpected, pollutant. As a result, we are trying to fit two peaks with a model capable of producing only one.  This results in five different posterior modes (Fig.~\ref{fig:multiple-spectral-modes}).
\begin{figure}[h]
  \begin{center}
    \includegraphics[width=\textwidth]{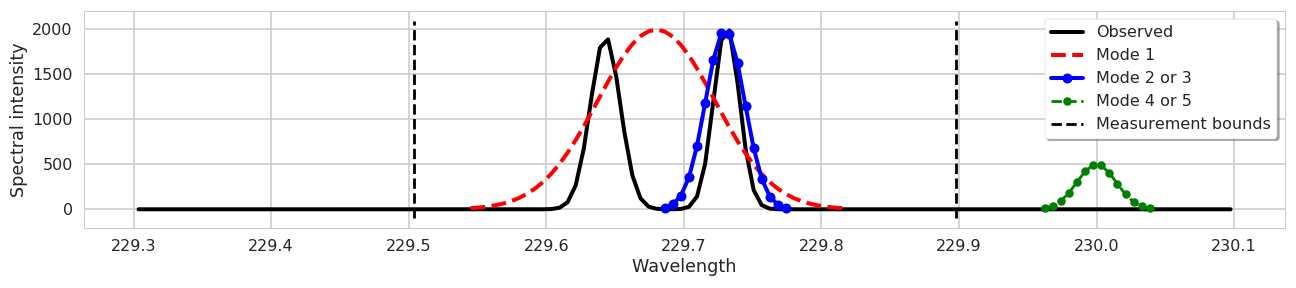}
  \end{center}
  \caption{
  {\bf Multiple modes:}  Observed 2-peak measurements, along with hypothetical \emph{posterior predictive} samples ($Y\sim p(y\g X)$ where $X\sim\posterior$) demonstrating different posterior modes.
  {\bf Mode 1:} If the temperature prior allows for large values, and/or the modeled noise is large enough, a local maximum is a hot plasma with bulk velocity near zero, giving rise to a single wide spectra that covers both peaks.
  {\bf Mode 2 or 3:} The next two modes arise if the temperature is a-priori small enough, with velocity values either positive or negative. These correspond to modeled spectral peaks that cover one of the measured peaks.
  {\bf Mode 4 or 5:} If the velocity prior is permissive, and the amplitude prior discourages tiny values, we have two more local maxima. These have velocity values so large that the modeled peak is shifted far to the left or right, outside the measured range of frequencies.
  }
  \label{fig:multiple-spectral-modes}
\end{figure}
Changing velocity determines the Doppler shift, and therefore the mode a sample is closest to. In a neighborhood of each mode, the log-posterior is concave. In this way our spectroscopy problem is similar to the toy model \eqref{align:multi-model-toy-model} with the number of velocity variables, $M$, equal to $N/3$.

\section{Review of Bayesian inference using MCMC}
\label{section:hmc-step-size}

Common to all Bayesian inversion is the goal of inferring a probabilistic description of an unobserved quantity $X\in\RN$, given observed $Y\in\RM$. This \emph{posterior}, $\posterior$, is proportional to the product of the \emph{prior} $\prior$ and the \emph{likelihood} $\likelihood$.  That is, $\posterior \propto \prior\likelihood$. The goal of this inference can be point estimates, such as $\Exp{X\g Y=y}$, or quantification of uncertainty, such as $\Var{X\g Y=y}$.

Markov Chains have been used since 1953 as a means to generate sequences of samples $X^1, X^2,\ldots$ from target distributions \cite{Robert2011-er}. These samples can be viewed by themselves, or used to approximate expectations $\Exp{f(X)}\approx S^{-1}\sum_{s=1}^S f(X^s)$.

\subsection{Statistical efficiency}
\label{subsection:statistical-efficiency}
Samples from the Markov Chain have two major deficiencies. First, only in the limit $k\to\infty$ can we say that $X^k\sim\pi$.  Often, a chain is initialized with a draw from the prior, and an initial \emph{burn-in} period is used to allow the chain to ``forget'' this initial position and move toward the \emph{typical set} of the posterior \cite{Betancourt2017-od}.
Second, draws are correlated.  Roughly speaking, we say that identically distributed samples $\left\{ X^1,\ldots,X^K \right\}$ have \emph{effective sample size} (ESS) $S$ if the variance of $K^{-1}\sum_{k=1}^K X^k$ is $\Var{X^1}/S$. In other words, if, for purposes of computing the mean, they are ``as effective'' as $S$ \iid samples.  Usually $S \leq K$.  The most straightforward computation of ESS assumes the $\{X^k\}$ are identically distributed but correlated draws from the target, with mean $\mu$ and variance $\nu^2$. Then, for dimensions $n=1,\ldots,N$, ESS is built from the $t-$lag autocorrelation $\rho_t$,
\begin{align*}
  \ess_n :&= \frac{K}{1 + 2\sum_{t=1}^\infty \rho_t},\qquad \rho_t := \frac{\Exp{(X_n^1 - \mu_n)(X_n^{1+t} - \mu_n)}}{\nu_n^2}.
\end{align*}
We are are usually most interested in the worst performing dimension, and therefore $\miness$.  Due to insufficient burn-in, or being stuck in local modes, the assumption that $X^k$ are draws from the target is often not satisfied. For that reason, our version of ESS will be lower if the mean differs between multiple chains. See appendix A of \cite{Vehtari2021-ai}. A more direct method of checking deviation between chains is the \emph{potential scale reduction}, or $\Rhat$ \cite{Vehtari2021-ai}.

\subsection{Computational efficiency}
\label{subsection:computational-efficiency}

In our experience, obtaining MCMC samples on production models is a computationally intensive task that can take prohibitively long.

Graphical processing units (GPUs) provide a significant speedup, as they excel at parallel (SIMD) operations. This allows concurrent running of multiple Markov Chains \cite{Lao2020-qi}. The slowdown as more chains are added is sub-linear, until GPU memory is exhausted.
We therefore only consider algorithms that take advantage of GPU parallelism.
In particular, our replica-exchange implementation uses the same number of leapfrog steps for each replica (see section \ref{subsubsection:selecting-number-of-leapfrog-steps-for-all-replicas}).

The majority of HMC computation time is spent in the numerical (leapfrog) integration of Hamilton's equations. Accordingly, we track the number of leapfrog steps per effective sample ($\miness$ actually) as a performance metric. This allows platform independent comparisons.

A secondary performance consideration is time spent compiling and optimization the computational graph \cite{Martin_Abadi2015-ju,Frostig_undated-rf}. This takes a fixed amount of time proportional to the graph complexity. In our model problem, where ill-conditioning and/or multi-modality was present, graph optimization was a small fraction of total runtime. In well-conditioned HMC sampling problems, avoiding many graph optimizations may be important.

\subsection{Description of HMC}
\label{subsection:description-of-hmc}
The Hamiltonian Monte Carlo (HMC) method was introduced in 1987 as ``Hybrid Monte Carlo'' for use in lattice field theory simulations \cite{Duane1987-eg}.  Since then, it has been recognized as an efficient alternative to random walk Metropolis-Hastings, well suited for higher dimensional problems.  Implementations are available for a variety of languages \cite{Carpenter2017-zd,Lao2020-qi,Salvatier2016-wx}. A comprehensive introduction to HMC can be found in \cite{Neal2011-gq}.

To sample $X\in\RN$, distributed according to a smooth (with respect to Lebesgue measure) density $\pi(x)$, HMC augments state space with a fictitious momentum $\xi\in\RN$.  This defines the joint density
\begin{align*}
  \pi(x,\xi) &\propto \exp\left\{ -H(x,\xi) \right\},\quad\mbox{where}\quad H(x,\xi) := - \log \pi(x) + \frac{\|\xi\|^2}{2},
\end{align*}
and $\|\xi\|$ is the Euclidean norm.
In the physics setting, the \emph{Hamiltonian} $H$ is total energy, whereas $-\log \pi(x)$, $\|\xi\|^2/2$ are potential and kinetic energies.
It is not uncommon to model the kinetic energy $\sqrt{\xi\cdot A^{-1}\xi}/2$, where $A\in\RNN$ is the \emph{mass matrix}.  However, as shown in section 4 of \cite{Neal2011-gq}, this is equivalent to the linear preconditioning $X\mapsto LX$ (where $A^{-1}=LL^T$) in conjunction with the Euclidean norm. Kinetic energy can be non-Gaussian e.g.~depend on position \cite{Girolami2011-wv}. See also \cite{Betancourt2017-od}.

Sampling proceeds by (a numerical approximation to) the following iteration from point $(x^j, \xi^j)$.
\begin{enumerate}
  \item Draw $\tilde \xi\sim\calN(0, I_N)$.
  \item Let $(x(t), \xi(t))$ be the time $t$ solution to the Hamilton's equations of motion:\, $\dot x = \xi$, $\dot \xi = \nabla \log \pi(x)$, with initial condition $(x^j, \tilde\xi)$.
  \item Set $(x^{j+1}, \xi^{j+1}) = (x(t), \xi(t))$, for integration time $t$.
\end{enumerate}
Each integration path lives on a single level set of the Hamiltonian. The resampling step $\tilde\xi\sim\calN(0, I_N)$ is necessary to jump between level sets, and is thus necessary for ergodicity \cite{Livingstone2016-pp}.

In practice, Hamilton's equations must be solved numerically over $\ell$ steps with step-size $h$. Denote this solution by $\Psi^\ell$.  The integration error means we can no longer just accept the move in step 3, which is replaced by a Metropolis correction:
\begin{align*}
  (x^{j+1}, \xi^{j+1}) &= \Psi^\ell,\quad\mbox{with probability } a( x^j, \xi^j \to \Psi^\ell),
\end{align*}
and
\begin{align*}
  (x^{j+1}, \xi^{j+1}) &= (x^j, \xi^j),\quad\mbox{with probability } 1 - a( x^j, \xi^j \to \Psi^\ell),
\end{align*}
for acceptance probability
\begin{align}
  \label{align:metropolis-step}
  a(x^j, \xi^j \to \Psi^\ell)
  :&= \min\left( 1, \exp\left\{ H(x^j, \xi^j) - H(\Psi^\ell) \right\} \right).
\end{align}
Since Hamilton's equations preserve the Hamiltonian, if numerical integration was perfect, $H(x^j, \xi^j) = H(\Psi^\ell)$ and every step would be accepted.  In practice, finite step size leads to some rejections and wasted effort.
We also note that for the Metropolis correction to be symmetric, the final momentum should be negated before evaluating $H$. This makes no difference since our Hamiltonian is symmetric in momentum.

The numerical integration is usually done with $\ell$ steps of the \emph{St\"ormer-Verlet} or \emph{leapfrog integrator}. This symplectic integrator ensures that the Hamiltonian does not diverge, provided $h$ is sufficiently small. See \cite{Hairer2006-gu}, theorem 8.1.

\subsection{HMC Step size scaling laws}
\label{subsection:hmc-step-size-scaling-laws}
Here we review existing scaling laws for HMC step size. These results always inform our choice of HMC step size and integration time. They are directly used in section \ref{section:sampling-with-correlations}.

The step size $h$, along with the number of leapfrog steps $\ell$, are two important parameters to choose.  Usually, $h$ is chosen to achieve some desired acceptance rate, and $\ell$ is set to the desired integration time divided by $h$.  If $h$ is too large, the average acceptance probability, $\rmP[\accept]$, will tend to zero. On the other hand, if $h$ is too small, then too many steps $\ell$ will be required, and computational effort will be wasted.
In \citep{Beskos2013-nd,Betancourt2014-kg}, $X\in\RN = \RK\times\cdots\RK$ with density $\pi(x_1,\ldots,x_{N/K}) = f(x_1)\cdots f(x_{N/K})$ is studied in the limit $N\to\infty$. This so-called \emph{\iid limit} yields the result that $h$ should be tuned until $0.6\leq \rmP[\accept] \leq 0.9$.

In the Gaussian case, we can extract more precise conclusions, resulting in useful analysis and algorithms that can be used, even in non-Gaussian problems.
Consider a target $X\sim\calN(0, C)$, where $C$ has eigenvalues $\lambda_1^2\geq\lambda_2^2\geq\cdots\geq\lambda_N^2>0$.  If we want the integration trajectories to travel a distance comparable to the largest scale $\lambda_1$, we must have $h\ell = O(\lambda_1)$. On the other hand, to avoid instability, we must have $h \leq 2\lambda_N$. This leads to a naive scaling $\ell = O(\lambda_1/\lambda_N)$. The ratio $\lambda_1/\lambda_N$ is the common, \emph{spectral}, condition number of any matrix $L$ such that $C=LL^T$.

The problem with the spectral condition number is that it takes into account only the largest and smallest scales. The largest scale does set the required integration length, but \emph{all} dimensions contribute to integration error.  A condition number taking these considerations into account was introduced in \cite{Langmore2019-lj}.
\begin{align}
  \label{align:kappa}
  \W(L) :&= \left( \sum_{n=1}^N\left( \frac{\lambda_1}{\lambda_N} \right)^4 \right)^{1/4}
  = \|L\|_2 \|L^{-1}\|_{S^4}.
\end{align}
Above, $\|\cdot\|_2$ is the \emph{spectral} norm, and $\|\cdot\|_{S^4}$ is the fourth \emph{Schatten} norm \cite{Horn1990-nf}. For matrix $A$ with singular values $\lambda_1\geq\cdots\lambda_N\geq0$, the $k^{th}$ Schatten norm is
\begin{align}
  \label{align:schatten-norm}
  \|A\|_{S^k} :&= \left[ \sum_{n=1}^N \lambda_n^k\right]^{1/k}.
\end{align}
In \cite{Langmore2019-lj}, it is shown that under some regularity conditions, a Gaussian density requires $O(\W)$ leapfrog steps for efficient sampling. $\W$ may be approximated using the largest scale $\lambda_1$, step size $h$, and average acceptance probability $\rmP[\accept]$.
\begin{align}
  \label{align:kappa-approximation-via-step-size}
  \W &\approx \frac{\lambda_1}{h} \, 2^{7/4}\sqrt{\Phi^{-1}\left( 1 - \frac{\rmP[\accept]}{2} \right)}.
\end{align}
Above, $\Phi$ is the normal cumulative distribution function. Thus, $\W$ is the number of steps needed to traverse the largest scale, $\lambda_1/h$, times a correction factor depending on acceptance probability. In non-Gaussian problems, we recommend using large (small) $\lambda_1/h$ as evidence of \emph{possible} of inefficient (efficent) sampling.

\subsection{Conductance in multi-modal HMC}
\label{subsection:conductance-basics}
Let $A\subset\RN$, and $c(A)$ the average probability of the chain escaping $A$ in one step. The \emph{conductance} of the chain, $\calC:=\min_{\pi(A)<1/2} \{c(A)\}$, quantifies its ability to escape local modes. In \cite{Mangoubi2018-hn}, it is shown that the conductance of HMC is no better than that of random walk Metropolis-Hastings. An upper bound is also given for conductance of HMC chains:
\begin{align}
  \label{align:conductance-bound}
  c(A) &\leq t\,\frac{1}{2} \frac{\int_{\p A}\pi(x)\dx}{\pi(A)}.
\end{align}
From \eqref{align:conductance-bound} it is clear that increasing integration time can increase conductance at most linearly. Since computational effort also increases linearly with $t$, this does not help overall.

Equation \eqref{align:conductance-bound} also shows that conductance can be constricted by a single low-probability (density) surface. For Gaussian noise models (e.g.~section \ref{subsection:model-problem}), this can lead to conductance decreasing exponentially in $1/\sigma^2$. Indeed, for the bi-modal normal $0.5\,\calN(-1, \sigma^2) + 0.5\,\calN(1, \sigma^2)$, the leading order term of conductance is proportional to $\exp\left\{ -1/(2\sigma^2) \right\}$. See \cite{Mangoubi2018-hn} theorem 3.

Tempering methods need to use a highest temperature $\Tmax$ such that the corresponding ``hottest density'' has non-vanishing conductance. See section \ref{subsubsection:selecting-replica-temperatures}.

\section{Preconditioning of the posterior covariance}
\label{section:sampling-with-correlations}

HMC (as well as standard Metropolis Hastings) works better when sampling from a unit Gaussian. The preconditioning techniques below sample a transformed variable $Z = F^{-1}(X)$, which (hopefully) looks more like a unit Gaussian. These samples $Z^j$ are transformed back into $X^j = F(Z^j)\sim p(x\g y)$.

\subsection{Transformation by diffeomorphism}
\label{subsection:transformation-by-diffeomorphism}
Here we write some relations involving pushforwards by diffeomorphisms. They will be needed in later sections.

Let us start with $X\sim p(\cdot\g y)$, and a diffeomorphism $F$, which transforms $X\mapsto Z = F^{-1}(X)$.  Equivalently, the density $p(\cdot\g y)$ is transformed by the \emph{pushforward} operator, $\pushfwd{F^{-1}}$, into $g$:
\begin{align}
  \label{align:pushfwd}
  g(z) &= (\pushfwd{F^{-1}}p(\cdot \g y))(z) := |\det(DF(z))|\, p(F(z)\g y).
\end{align}
Above, $DF$ is the matrix of partial derivatives, $(DF)_{ij} = \p F_j/\p z_i$.
Using HMC, we sample from the transformed density $g$, producing $Z^1,\ldots,Z^K$.  Transforming back, $X^k := F(Z^k)$, and we have samples from $p(\cdot\g y)$ as desired.

In the Gaussian case, $p(\cdot\g y)\sim\calN(\mu, C)$, with $C=LL^T$, the linear preconditioner induced by a matrix $F$ transforms $L\mapsto F^{-1}L$, and likewise the covariance and $\W$ as
\begin{align}
  \label{align:linear-preconditioning-transformations}
  \begin{split}
    LL^T & \mapsto (F^{-1}L)(F^{-1}L)^T = (F^{-1})LL^T(F^{-1})^T \\
    \W(L) & \mapsto \|F^{-1}L\|_2\|(F^{-1}L)^{-1}\|_{S^4}.
  \end{split}
\end{align}

\subsection{Preconditioning by prior-reparameterization}
\label{subsection:preconditioning-by-reparameterization}
If the support of $\posterior$ is bounded, sampling directly from it will suffer from boundary issues. It is therefore standard practice to transform $X\to Z := F^{-1}(X)$ (as in section \ref{subsection:transformation-by-diffeomorphism}) such that $\supp(Z)=\RN$. We know of no existing work analyzing the effects of this reparameterization on conditioning, so we include it here.

In our case, the prior $p(x)$ is a transformation of a standard Gaussian by a diffeomorphism $G$.  That is, $p(x) = (\pushfwd{G}\phi)(x)$.  This means $Z = G^{-1}(X)$ is a useful transformation.  If the prior results in difficult posterior covariance (see e.g.~section \ref{subsection:covariance-structures}), using $G$ as a preconditioner will often improve conditioning. To that end, note that
\begin{align}
  \label{align:reparameterization}
  \begin{split}
      \pushfwd{G^{-1}} p(x\g y)
    &= p(G(x)\g y)|\det(DG(x))|
    \propto \phi(x) p_Y(y\g G(x)).
  \end{split}
\end{align}
Thus, preconditioning with $G$ is equivalent to a reparameterization that uses a standard Gaussian prior, and inserts $G$ inside the likelihood (probably achieved by using $G(x)$ in the forward model). In our work, we sometimes explicitly preconditioned with $G$, and at other times reparameterized a unit normal.  Not only are these mathematically equivalent, but, due to caching in TensorFlow probability, they are computationally equivalent ($G\circ G^{-1}$ is replaced by the identity) \cite{Dillon2017-zq}. Doing one of these is necessary, as it removes much of the nonlinearity and allows the methods of section \ref{subsection:linear-preconditioning} to work efficiently.

Organizing code around reparameterizations has some advantages: First, the benefits of preconditioning by the prior are realized without the programmatic complexity of specifying a preconditioner. Second, the function $G$ does not need to be a diffeomorphism; see \cite{Dikovsky2021-cm} for an example.  On the other hand, there are benefits to letting the prior encode the state directly. First, in this setup, the MCMC variables being sampled are the state variables you care about. Second, one can use a prior \emph{not} easily described as a transformation of a Gaussian.

The sampling benefits/degradation of prior-reparameterization will vary. Consider the Gaussian example \eqref{align:gaussian-inverse-problem}. If $C_{pr} = L_{pr}L_{pr}^T$, then the reparameterized posterior covariance becomes
\begin{align}
  \left[ I + \sigma^{-2}(AL_{pr})^T(AL_{pr}) \right]^{-1}.
  \label{align:gaussian-posterior-covariance-reparam}
\end{align}
We see in figure \ref{fig:gaussian-posterior-covariance-spectra-reparam} that reparameterization can help when prior covariance is the major contributor to $\W$, but can hurt if low noise is the major source of small eigenvalues.

\begin{figure}[h]
  \begin{center}
    \includegraphics[width=\textwidth]{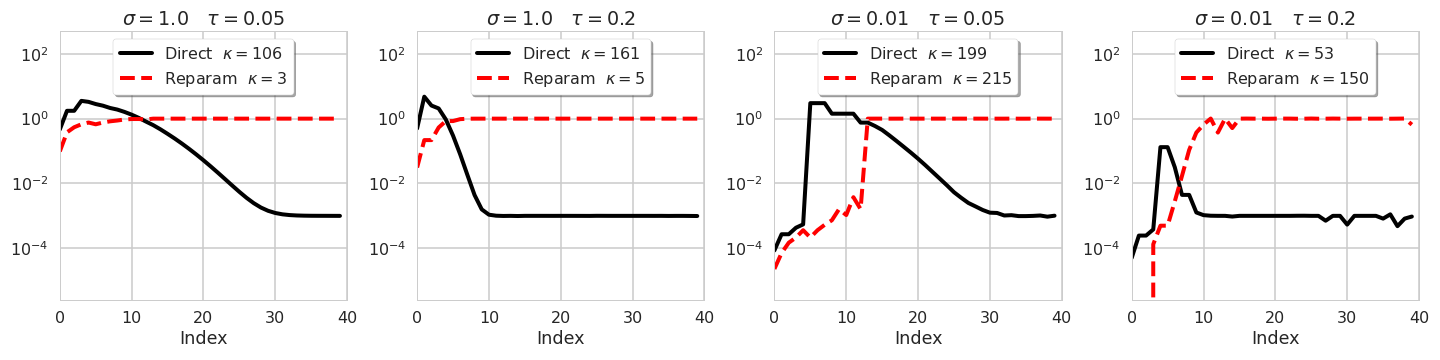}
  \end{center}
  \caption{
  {\bf Covariance spectra:} Eigenvalues, indexed by correlation with the prior eigenvectors.
  From the posterior before reparameterization \eqref{align:gaussian-posterior-covariance}, and after \eqref{align:gaussian-posterior-covariance-reparam}. Reparameterization ``lifts'' the high frequency posterior scales.
  Different noise levels $\sigma$ and prior correlation lengths $\tau$ result in reparameterization helping or hurting.
  In all cases, $M=20$, $N=40$, and the prior shrinkage $\delta=0.001$.
  The minimal condition number $\W$ (see \eqref{align:kappa}) is $40^{1/4}\approx 2.5$.
  }
  \label{fig:gaussian-posterior-covariance-spectra-reparam}
\end{figure}

\subsection{Linear preconditioning for Gaussian problems}
\label{subsection:linear-preconditioning}
Here we examine linear preconditioning of Gaussian distributions. This simplification allows for precise results and inspires techniques. These techniques are applied to non-Gaussian problems in section \ref{subsection:nongaussian-preconditioning-numerics}.

The setup here is like other diffeomorphisms (section \ref{subsection:transformation-by-diffeomorphism}), except we assume $X\sim \calN(0, C)$ is Gaussian, and the preconditioning transformation $F\in\RNN$ is linear. If $FF^T$ is a scalar multiple of $C$, then \eqref{align:linear-preconditioning-transformations} shows the transformed covariance is a multiple of the identity, which minimizes $\W$.
Subsequent sections will use approximations of this factor built from samples obtained during burn-in.

\subsubsection{Diagonal preconditioning}
A simple approximation of a covariance factor is the diagonal matrix $\Dhat$ made with sample standard deviations. This approximates $D$, the diagonal matrix of \emph{actual} standard deviations. The next three propositions are original, and illustrate why we like this preconditioner.

First, preconditioning with $D$ should work well in the diagonally dominant case.
The proof follows from the Gershgorin circle theorem \cite{Horn1990-nf}.
\begin{proposition}
  \label{proposition:diag-preconditioning-for-diagonally-dominant}
  Suppose $C=LL^T$ is diagonally dominant with, for every $i$, $\sum_{j\neq i}|C_{ij} / C_{ii}| \leq \delta < 1$.
  Then,
  \begin{align*}
    \W(D^{-1}L) \leq N^{1/4}\sqrt{\frac{1 + \delta}{1 - \delta}}.
  \end{align*}
\end{proposition}

Second, $D$ is close to the ideal diagonal preconditioner, especially in the near-diagonal case.
\begin{proposition}
  \label{proposition:equilibriation-near-ideal}
  Let $D_{opt}$ be a preconditioner minimizing $\W(G^{-1}L)$ over all diagonal matrices $G$. Then,
  \begin{align*}
  \W(D^{-1}L)\leq \sqrt{N}\,\W(D_{opt}^{-1}L).
  \end{align*}
  Furthermore, if at most $K$ entries in each row of $LL^T$ are nonzero, then
  \begin{align*}
    \W(D^{-1}L)\leq \sqrt{K}\,\W(D_{opt}^{-1}L).
  \end{align*}
\end{proposition}
The proof is found in section \ref{appendix:subsection:equilibration}.
A corollary is that preconditioning with $D$ can hurt by at most a factor of $\sqrt{N}$ (or $\sqrt{K}$).

Third, in practice, we use the diagonal matrix of sample standard deviations, $\Dhat$, rather than $D$. If we use $S$~\iid samples from $\calN(0, LL^T)$, then $S$ need only grow logarithmically with $N$.
\begin{proposition}
  \label{proposition:sample-stddevs-good-enough}
  Given $\eps, p\in(0, 1)$,
  \begin{align*}
    \W(\Dhat^{-1}L) &\leq \W(D^{-1}L)\, \sqrt{\frac{1 + \eps}{1 - \eps}},
  \end{align*}
  with probability $p$, as soon as the number of~\iid samples $S$ satisfies
  \begin{align*}
    S &\geq \frac{25}{\eps^2}\log \left( \frac{3N}{p} \right).
  \end{align*}
\end{proposition}
\begin{proof}
  We have the distributional equality, $\Dhat \equalindistribution \sqrt{(I + \Delta)}D$, where $\Delta\in\RNN$ is diagonal,
  \begin{align*}
    \Delta_{nn} :&= \frac{1}{S}\sum_{s=1}^S Z_{s,n}^2 - 1 - \left( \frac{1}{S}\sum_{s=1}^S Z_{s,n}\right)^2,
  \end{align*}
  and $Z_{s,n}$ are~\iid~normal variates.  Then, using lemma \ref{lemma:schatten-norm-monotonicity},
  \begin{align*}
    \W(\Dhat^{-1} L)
    &= \|\Dhat^{-1}L\|_2 \|L^{-1}\Dhat\|_{S^4}
    \leq \sqrt{\frac{1 + \max_n\{\Delta_{nn}\}}{1 + \min_n\{\Delta_{nn}\}}}\, \W(D^{-1}L).
  \end{align*}
  The proof will be complete once we show that our condition on $S$ implies $\max_n|\Delta_{nn}| \leq \eps$ with probability less than $p$. This follows from lemma \ref{lemma:chi-square-bound} and the fact that $\rmP[\max_n |\Delta_{nn}| \geq \eps] \leq N\, \rmP[|\Delta_{11}| \geq \eps]$.
\end{proof}

\subsubsection{Full covariance preconditioning}
\label{subsubsection:full-cov-preconditioning}
By \emph{full covariance preconditioning}, we mean starting with the sample covariance, $\Chat$, factorizing as $\Chat = \Lhat\Lhat^T$, then preconditioning with $\Lhat$.
After stating two results on full covariance preconditioning, we discuss a scheme for implementing it. The results and the scheme were discussed in our previous work, \cite{Langmore2019-lj}. The scheme is discussed in more detail and implemented here for the first time (section \ref{subsection:nongaussian-preconditioning-numerics}).

Remarkably, if the samples are independent, the condition number does \emph{not} depend on the true covariance.
\begin{lemma}
  \label{lemma:preconditioned-is-inverse-wishart}
  Suppose $(X^1,\ldots,X^S)$ are~\iid~samples of $X\sim\calN(0, C)$, and we precondition sampling of $X$ with the $S$-sample factor $\Lhat$.  Then, the preconditioned $\W$ follows the law of $\W(B)$, for $BB^T\sim\InverseWishart{S}{N}$.
\end{lemma}
In the high dimensional limit, $\W$ for inverse Wishart matrices has a simple expression.
\begin{proposition}
  \label{proposition:inv-wishart-kappa}
  If $BB^T\sim\InverseWishart{N}{S}$, and $N\to\infty$ with $S / N\to \omega\in(1, \infty)$, then
  \begin{align*}
    \frac{\W(B)}{N^{1/4}} &\to \frac{(1 + \omega^{-1})^{1/4}}{1 - \omega^{-1/2}}
  \end{align*}
  almost surely.
\end{proposition}
See figure \ref{fig:kappa-estimates} for a visualization of \ref{lemma:preconditioned-is-inverse-wishart} and proposition \ref{proposition:inv-wishart-kappa}. Due to the ``universality'' of random matrices, we expect these results to hold for linear transformations of a wide variety of~\iid~random variables \cite{Bai1999-dh,Bose2010-wh}.
\begin{figure}[h]
  \centering
  \includegraphics[width=0.66\textwidth]{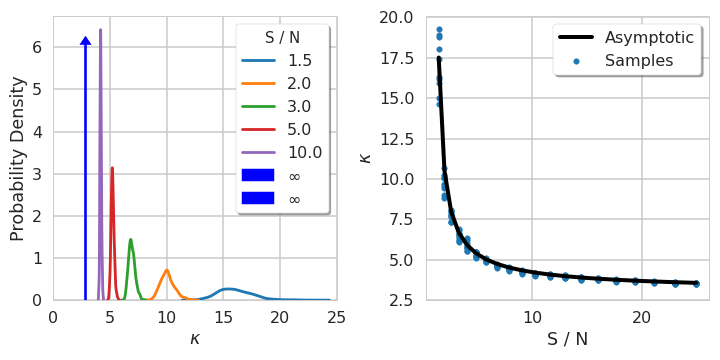}
  \includegraphics[width=0.33\textwidth, trim=390 400 0 0, clip=true]{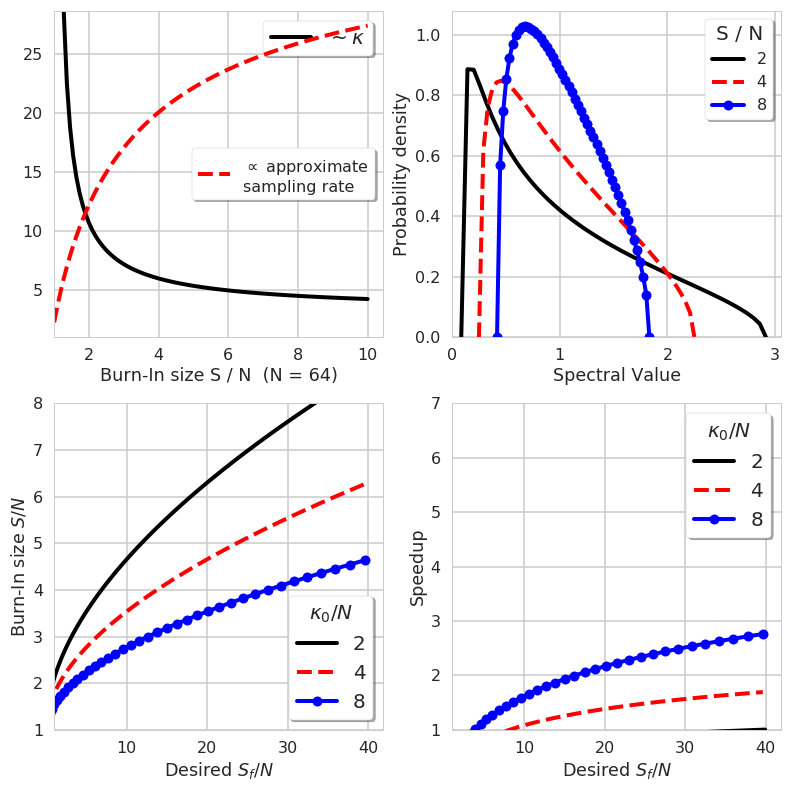}
  \caption{
  \label{fig:kappa-estimates}
  \textbf{Density and asymptotic} $\W(B)$, when $BB^T\sim\InverseWishart{N}{S}$ with $N=64$.
  \textbf{Left:} Density plots of sample values of $\W(B)$ for different $S/N$. $\W(B)\to N^{1/4}$ as $S/N\to\infty$.
  \textbf{Center:} Asymptotic estimate (from proposition \ref{proposition:inv-wishart-kappa}) vs.\ samples of $\W$. Once $S / N\approx 20$, $\W$ is close to the ideal value of $N^{1/4}\approx2.8$.
  \textbf{Right:} The Mar\u{c}enko-Pastur density, which is the limiting spectral density of $\Wishart{N}{S}$.
  }
\end{figure}

These results lead to a useful preconditioning scheme.
Let $\W_0$ be the condition number before preconditioning, and $\W_S$ be the condition number after preconditioning with $S$ \iid samples.  Assuming the sampling rate is proportional to $1 / \W$, the time to obtain $S_f$ ``final'' samples is proportional to  $S\W_0 + S_f \W_S$. On the other hand, without preconditioning, the time is proportional to $S_f\W_0$. This means the speedup from preconditioning is
\begin{align}
  \label{align:optimal-schedule-speedup}
  \frac{S_f\W_0}{S\W_0 + S_f\W_S}.
\end{align}
Estimating $\W_0$ using \eqref{align:kappa-approximation-via-step-size}, and using proposition \ref{proposition:inv-wishart-kappa} as an expression for $\W_S$, we can compute speedup for various $S$. If the maximal speedup (using $S^\ast$ samples) is $>1$, we proceed with drawing $S^\ast$ burn-in samples, precondition with  $\Lhat$, then draw our $S_f$ final samples.  If not we draw $S_f$ samples without preconditioning.

Mentioning some practicalities is in order. Burn-in samples obtained using standard HMC are far from independent. As a remedy, we use the No-U-Turn Sampler (NUTS) \citep{Hoffman2014-vt} to obtain the $S^\ast$ preconditioning burn-in samples, and stop sampling when the mean (across dimensions) effective sample size is $S^\ast$. In our experience, obtaining NUTS samples takes around 3x longer than standard HMC samples. This happens since NUTS sampling involves doubling the trajectory length and resampling within these long trajectories.  Moreover, the additional preconditioning stage requires another step size adaptation stage.  We therefore replace $\W_0$ with $4 \W_0$ in the denominator of \eqref{align:optimal-schedule-speedup}. This leads to algorithm \ref{algorithm:kappa-sampling}. See also plots of samples in different stages in figure \ref{fig:traces}.

\begin{algorithm}[h]
  \label{algorithm:kappa-sampling}
  \SetAlgoLined
  Initialize 20 chains by sampling from the prior\;
  Start $h$ small enough so $\rmP[\accept]\approx1.0$, then adapt $h$ until $\rmP[\accept]\approx0.9$. Use number of leapfrog steps $\ell=5$.  When done, set $\ell=(1/h) (\pi/2)$\;
  Draw stage 1 samples. Use them to compute the largest scale $\lambda_1$, then set $\ell = (\lambda_1/h)(\pi/2)$\;
  Draw stage 2 samples. Use them to re-compute $\lambda_1$ and $\rmP[\accept]$. Compute $\W_0$ using \eqref{align:kappa-approximation-via-step-size} and maximal speedup using \eqref{align:optimal-schedule-speedup}\;
  \If{maximal speedup $ > 1$}
  {
  \While{$N^{-1}\sum_{n=1}^N\{ESS_n(\mbox{stage 3 samples})\} < S^\ast$}{
    Draw more stage 3 samples using NUTS\;
  }
    Precondition using the stage 3 sample covariance\;
    Adapt step size until $\rmP[\accept]\approx0.9$, and set $\ell=(1/h)(\pi/2)$\;
  }
  \While{$\miness(\mbox{final stage samples}) < S_f$}{
    Draw more \emph{final} stage samples\;
  }
  \KwResult{$S_f$ ``final'' samples}
  \caption{Sampling stages for unimodal problems. Abandon and restart using REMC if $\Rhat$ fails to reduce fast enough.}
\end{algorithm}
Each step size adaptation in algorithm \ref{algorithm:kappa-sampling} is done via an iterative scheme \cite{Andrieu2008-xv}, invoked after preemptively adjusting step size using \eqref{align:kappa-approximation-via-step-size}. Step size adjustment comprises around 30\% of runtime. This could often be shorter, but this longer adaptation makes the algorithm more robust to stuck chains.

\subsection{Application to a weakly non-Gaussian problem}
\label{subsection:nongaussian-preconditioning-numerics}
Here we compare preconditioning schemes as applied to the shell model from section \ref{subsection:model-problem}.
This problem is non-Gaussian. In particular, the noise level depends on the signal, the temperature and amplitude are constrained to be positive via a Softplus, the coordinate system is shifted, and emissivity is a nonlinear function. This is still ``weakly'' non-Gaussian, since observed skew and kurtosis levels of transformed samples $Z$ were close to that of a Normal.

The schemes compared are referred to as ``full'', ``diag'', and ``none''.  ``Full'' uses algorithm \ref{algorithm:kappa-sampling}. ``Diag'' uses algorithm \ref{algorithm:kappa-sampling} but skips the NUTS sampling and uses diagonal rather than sample covariance preconditioning. ``None'' does not precondition.  Code was run on Tesla P100 GPUs. 

20 plasmas to reconstruct were drawn from the prior. The reconstruction model used a variety of noise levels from $\sigma=1.25$ to 15.  Each reconstruction was run until $\miness=S_f$, for $S_f\in\left\{ 400, 1600, 6400 \right\}$.  A total of 282 reconstruction configurations were attempted for each of the three schemes.  Eleven configurations were thrown out, since at least one model failed to reduce $\Rhat$ fast enough. This is usually the result of being stuck in a local mode due to poor initialization. If this happens in production, our algorithm re-starts with REMC.
The traces in figure \ref{fig:traces} help visualize different stages of algorithm \ref{algorithm:kappa-sampling}.
Figure \ref{fig:kappa-estimation-in-practice} shows $\W$ correlates well with sampling efficiency, and can be predicted from proposition \ref{proposition:inv-wishart-kappa}, even in this non-Gaussian problem.
Figure \ref{fig:runtime-correlation-study} shows that the full preconditioner significantly speeds up sampling, at the cost of a more expensive burn-in. Diagonal preconditioning helps only a little.

\begin{figure}[h]
  \begin{center}
    \includegraphics[width=0.49\textwidth]{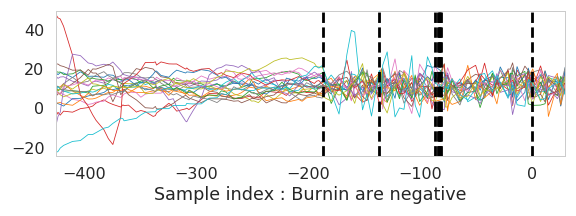}
    \includegraphics[width=0.49\textwidth]{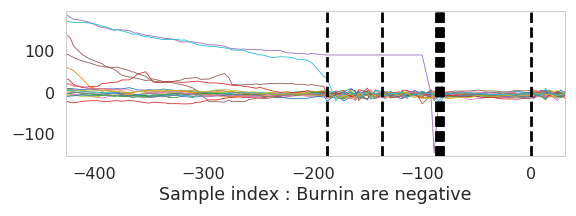}
  \end{center}
  \caption{
  {\bf Traces of algorithm \ref{algorithm:kappa-sampling}:} Plots of coordinate samples ($Z^s_n$) are \emph{the} most important diagnostic tool.
  {\bf Left:} Stages of algorithm \ref{algorithm:kappa-sampling} are divided by dotted lines. Leftmost is the initial step size adaptation phase. Using a large number of samples here allows chains to reach the typical set. The low $\ell$ value and low initial $h$ means sampling proceeds slowly, but is likely to be stable. Second from the left are the ``stage 1'' samples. These have $\ell$ large enough to get a reasonable estimate of $\lambda_1$, and also allow chains more time to reach the typical set. Next are the ``stage 2'' samples, used to compute $\W_0$ and the number of preconditioning samples needed, $S^\ast$.  Next are 4 very closely spaced dashed lines, within which NUTS sampling was used to obtain $S^\ast$ effective samples. These are used to form the sample covariance factor $\Lhat$ used for preconditioning. Second stage from the right is step size adaptation done after preconditioning. The final stage, starting at 0, includes the first 25 ``final'' samples.
  {\bf Right:} Same stages, in a case where preconditioning hurt. The reason here is that some chains did not reach the typical set before preconditioning samples were taken. This led to a bad preconditioner.
  }
  \label{fig:traces}
\end{figure}
\begin{figure}[h]
  \begin{center}
    \includegraphics[width=\textwidth]{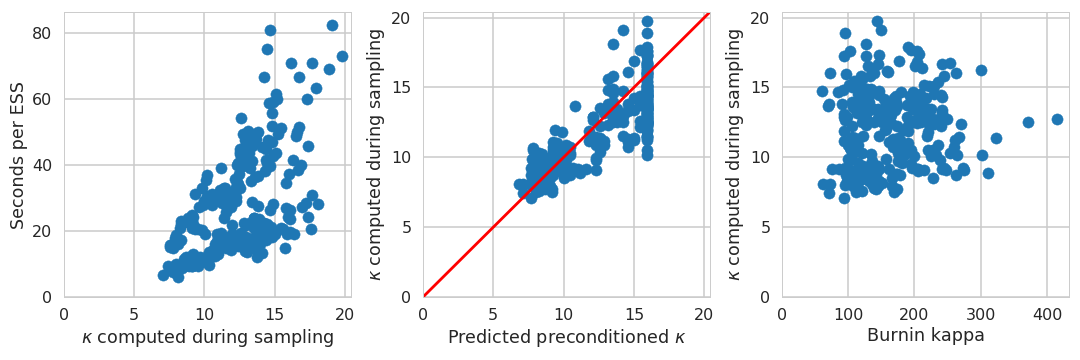}
  \end{center}
  \caption{
  Using \eqref{align:kappa-approximation-via-step-size}, we compute $\W$ at various sampling stages to show our formulas apply, even in non-Gaussian problems.  Experiments when the burn-in $\W$ was huge (upper 2\%) are not shown, as these distort the plot. Preconditioning often made the situation worse for these.
  {\bf Left:} The relationship between $\W$ and the seconds required for effective samples is somewhat close to linear.  This validates \eqref{align:optimal-schedule-speedup} as a measure of speedup from preconditioning.
  {\bf Center:} The predicted post-preconditioning value of $\W$ matches nicely with the actual value obtained by preconditioning, validating proposition \ref{proposition:inv-wishart-kappa}.
  {\bf Right:} The post-preconditioning $\W_S$, is plotted against the burn-in value, $\W_0$. This shows a significant reduction in $\W$ due to algorithm \ref{algorithm:kappa-sampling}.
  }
  \label{fig:kappa-estimation-in-practice}
\end{figure}

\begin{figure}[h]
  \begin{center}
    \includegraphics[width=0.32\textwidth]{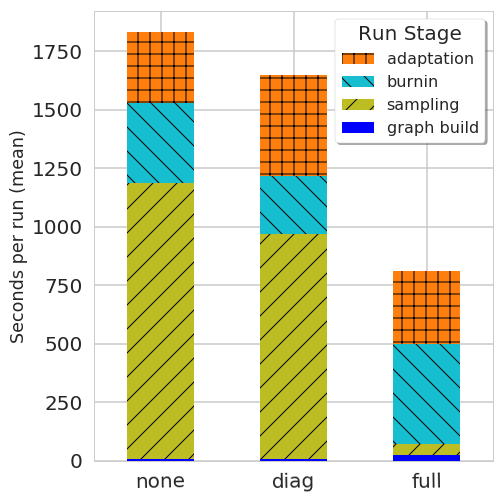}
    \includegraphics[width=0.32\textwidth]{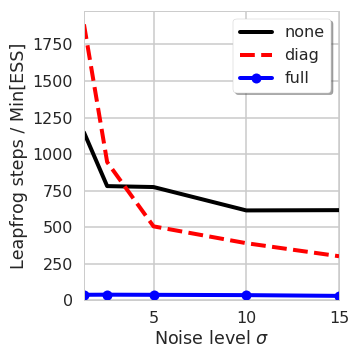}
    \includegraphics[width=0.32\textwidth]{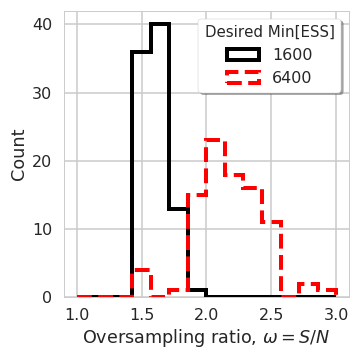}
  \end{center}
  \caption{
  {\bf Preconditioner comparison:} From the study outlined in section \ref{subsection:nongaussian-preconditioning-numerics}.
  {\bf Left:} Average runtime breakdown ($S_f=1600$) shows that while the burn-in time is longer for full preconditioning, the sampling time is significantly shorter.
  {\bf Center:} Full covariance preconditioning results in sampling efficiency around 36, which is around 30x better than not preconditioning.
  {\bf Right:} Histogram of oversampling ratio $\omega=S/N$ selected by \eqref{align:optimal-schedule-speedup} for two values of desired $\miness$. When the desired $\miness$ is larger, a larger $\omega$, is selected.
  }
  \label{fig:runtime-correlation-study}
\end{figure}

\section{Tempering to sample with multi-modality}
\label{section:sampling-with-multi-modality}
Best practices for sampling from multi-modal distributions are not as easy to come by as for their unimodal counterparts. For example, linear preconditioning usually does not help. A popular family of techniques involves using a number of modifications of the target, each \emph{tempered} by temperature $T$. The terminology and history is rooted in statistical mechanics \cite{Swendsen1986-wl}. These techniques make use of the fact that if $\pi$ is a probability density, and temperature $T>1$, the density proportional to $\pi^{1/T}$ will have lower peaks and higher troughs. Hence, it will be better able to jump between modes.

\subsection{Replica Exchange Monte Carlo (REMC)}
\label{subsection:remc-intro}
This section reviews REMC (also known as \emph{parallel tempering}). Related techniques, such as \emph{annealed importance sampling}, also deserve consideration \cite{Neal2001-mg}. 

Given posterior $\posterior \propto \prior\likelihood$, and sequence of temperatures $1=T_1<T_2<\cdots<T_R\leq\infty$, we form the \emph{replica} densities $\pi_r$ in one of two ways.
\begin{align}
  \label{align:tempered-densities}
  \pi_r(x) &\propto \left\{ 
  \begin{matrix}
    \prior^{1/T_r} \likelihood^{1/T_r},&\quad\mbox{posterior tempering (requires $T_R<\infty$)},\\
    \prior\likelihood^{1/T_r},&\quad\mbox{likelihood tempering}.
  \end{matrix}
  \right.
\end{align}

To gain intuition, consider the unimodal example where the prior $p(x)\sim\calN(0, I)$ and the likelihood $p(y\g x)\sim\calN(\mu, \Gamma)$. The posterior covariance after tempering with $T$ will be
\begin{align}
  \label{align:tempered-posterior-covariance-gaussian}
  \Gamma_{post}(T) :&= \left\{ 
  \begin{matrix}
    T \left[ I + \Gamma^{-1} \right]^{-1}, \quad\mbox{posterior tempering},\\
    \left[ I + T^{-1}\Gamma^{-1} \right]^{-1}, \quad\mbox{likelihood tempering}.
  \end{matrix}
  \right.
\end{align}
Posterior tempering increases posterior variance without changing the shape or condition number. Likelihood tempering distorts the posterior covariance to make it look like the prior.
See figure \ref{fig:tempering-scatterplots}.
\begin{figure}[h]
  \begin{center}
    \includegraphics[width=0.9\textwidth]{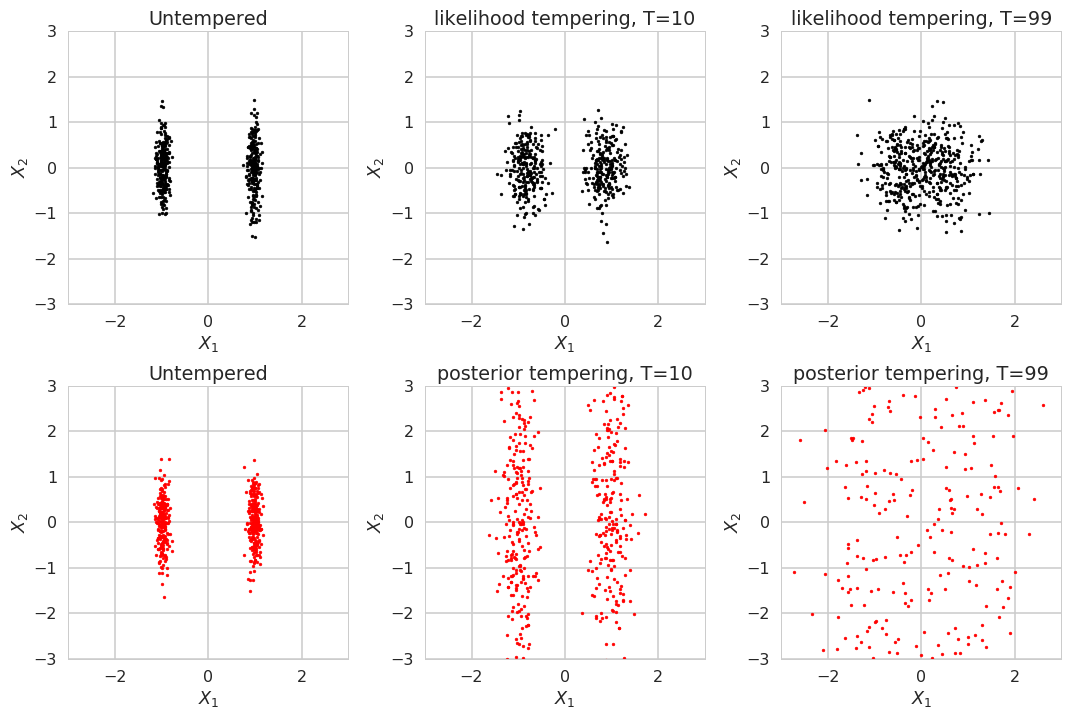}
  \end{center}
  \caption{
  {\bf Likelihood and posterior tempering:} Sampling from a tempered bi-modal normal (an $M=1$ version of \eqref{align:multi-model-toy-model}).
  {\bf Top:} Likelihood tempering means the hottest ($T=99$) replica samples come from a (nearly) isotropic Gaussian.
  {\bf Bottom:} Posterior tempering means the $T=10$ replica samples come from a (nearly) bi-modal normal, with each mode being (nearly) as poorly conditioned as the posterior ($T=1$) modes.
  }
  \label{fig:tempering-scatterplots}
\end{figure}

Together, these form the joint density $\pi(x_1,\ldots,x_R) := \pi_1(x_1)\cdots \pi_R(x_R)$. Samples from $(X_1^{k},\ldots,X_R^{k}) \sim\pi$ are generated, but only the target samples $X_1^{k}\sim\pi_1 = p(\cdot\g y)$ are kept.
REMC repeats two alternating steps. In the \emph{exploration} step, each replica progresses independently, $X^{k-1}_r \to X^k_r$. HMC, or another sampling method, can be used here. In the \emph{communication} step, a number of \emph{swaps} are proposed between adjacent replicas. For example, the $(1,2)$ swap proposes that replicas 1 and 2 exchange position; $X^{k+1}_1 = X^k_2$ and $X^{k+1}_2 = X^k_1$.
The set of swaps proposed at each turn are either the set of even swaps $\left\{ (1,2), (3, 4),\ldots \right\}$, or odd swaps $\left\{ (2, 3), (4, 5),\ldots \right\}$. In both cases, each swap is independently accepted or rejected according to the standard Metropolis-Hastings criteria. For example,
\begin{align}
  \label{align:swap-prob-12}
  \begin{split}
    \rmP[\swap_{(1,2)}\g x_1, x_2]
    &= \min\left\{ 1, \alpha_{(1,2)(x_1, x_2)} \right\},\\
    \alpha_{(1,2)} :&= \frac{\pi(x_2,x_1,x_3,\ldots)}{\pi(x_1,x_2,x_3,\ldots)}
    = \frac{\pi_1(x_2)\pi_2(x_1)}{\pi_1(x_1)\pi_2(x_2)}.
  \end{split}
\end{align}

The computational cost of swapping is negligible compared with leapfrog integration.  We therefore propose swaps between every exploration step. In most of the literature, swaps are proposed using a \emph{stochastic even-odd} (SEO) scheme, whereby a coin flip chooses between even and odd swaps. More recently, deterministically alternating between even and odd swaps was proposed \cite{Okabe2001-cy}. This \emph{deterministic even odd} scheme (DEO) scheme has superior scaling characteristics \cite{Syed2022-ag}. We use DEO in all experiments.

\subsection{Swapping, dimension laws, and under-constrained problems}
\label{subsection:swapping}
Effective REMC requires the mean swap probabilities,
\begin{align*}
  \swapprob = \Exp{\rmP[\swap_{(r, r+1)}\g X_r, X_{r+1}]},
\end{align*}
be large enough to allow information from the hottest replica (the replica using the highest temperature) to make its way to the target. This section shows that, unlike conductance, swap probability is strongly related to dimension.

\subsubsection{Existing fundamental results}
\label{subsubsection:existing-fundamental-results}
In higher dimensions, samples concentrate in a thin neighborhood of the \emph{typical set} \cite{Betancourt2017-od}.  For example, if $\pi$ is log-concave, the typical set is $\{x\st \log\pi(x) = \Exp{\log\pi(X)}\}$, and the neighborhood grows (relatively) thinner as dimension increases \cite{Bobkov2011-mj}.
Since swapping of replicas $r$ and $r+1$ must lead to valid samples from their respective densities, they must swap about as often as these neighborhoods overlap.  This overlap is made explicit by
\begin{proposition}
  \label{proposition:swap-probs-as-typical-set-overlap}
  Let $X_r$ be a sample from replica $\pi_r$ defined by \eqref{align:tempered-densities}.
  The mean swap probability can be written in terms of the untempered posterior/likelihood:
  \begin{align*}
    \rmP[\swap_{(r, r+1)}]
    &= \left\{
    \begin{matrix}
      2\,\rmP[p(X_r \g y) < p(X_{r+1}\g y)],&\quad\mbox{posterior tempering},\\
      2\,\rmP[p(y \g X_r) < p(y\g X_{r+1})],&\quad\mbox{likelihood tempering}.
    \end{matrix}
    \right.
  \end{align*}
\end{proposition}
\begin{proof}
  The proof has likely been shown many times before. See e.g.~\cite{Kofke2002-my}, for the case of posterior tempering.
  Consider likelihood tempering, and write $\pi_r(x)\propto \exp\left\{ -V(x)/T_r - V_0(x) \right\}$.
  The Metropolis criteria \eqref{align:swap-prob-12} gives
  \begin{align}
    \label{align:mean-swap-prob-integral}
    \rmP[\swap_{(r, r+1)}]
    &= \int \min\left\{ \pi_r(x_r)\pi_{r+1}(x_{r+1}),\, \pi_r(x_{r+1})\pi_{r+1}(x_r) \right\}\dx_r\dx_{r+1}.
  \end{align}
  Since $T_r < T_{r+1}$,
  \begin{align*}
    \pi_r(x_r)\pi_{r+1}(x_{r+1}) < \pi_r(x_{r+1})\pi_{r+1}(x_r)
    & \Longleftrightarrow\frac{V(x_r)}{T_r} + \frac{V(x_{r+1})}{T_{r+1}} > \frac{V(x_r)}{T_{r+1}} + \frac{V(x_{r+1})}{T_r} \\
    &\Longleftrightarrow V(x_r) > V(x_{r+1}).
  \end{align*}
  This leads us to split the integral \eqref{align:mean-swap-prob-integral} up over regions $\left\{ V(x_r) > V(x_{r+1}) \right\}$ and $\left\{ V(x_r) < V(x_{r+1}) \right\}$. A switch of the dummy variables $x_r$, $x_{r+1}$ in the second shows that both integrals are equal to $\rmP[V(X_r) > V(X_{r+1})]$, which gives to the desired result. The case of posterior tempering is similar.
\end{proof}

For REMC to work well, information must propagate from the hottest replica (replica using the highest temperature) to the target.
To study this, one can keep track of the \emph{index process} of temperatures. For example, chain $k$ may start by sampling from $\pi_R$, then swap and sample from $\pi_{R-1}$, then $\pi_{R-2}$, $\pi_{R-1}$ and so on. The corresponding indices would be $(R, R-1, R-2, R-1,\ldots)$. A \emph{round trip} occurs when a chain starts at index $k$, reaches $R$, then $1$, then back to $k$.  The average number of round trips, starting from all replicas, after $S$ swap attempts, is the \emph{round trip rate}.
To derive round trip rates for likelihood tempering when $T_R=\infty$, \cite{Syed2022-ag} makes three assumptions: First, \emph{stationarity}, $X_r\sim\pi_r$, which is reasonable after burn-in. Second, \emph{efficient local exploration} (ELE). ELE means that, if $X\sim\pi_r$, and $X'$ is the result of local exploration (e.g.~HMC integration) starting from $X$, then the potential energy is independent. In the case of posterior tempering, this means $\log[p(X\g y)]$ and $\log[p(X'\g y)]$ are independent, and for likelihood tempering, $\log[p(y\g X)]$ and $\log[p(y\g X')]$ are independent. Third, they assume integrability of the cubed log likelihood. This leads to round trip rates for the SEO and DEO swapping schemes:
\begin{align}
  \label{align:round-trip-rates}
  \tau_{SEO} &= \frac{1}{2 R + 2 \gamma},
  \quad
  \tau_{DEO} = \frac{1}{2 + 2 \gamma},
\end{align}
where $\gamma$ is the \emph{schedule inefficiency}
\begin{align*}
  \gamma :&= \sum_{r=1}^{R-1} \frac{1 - \rmP[\swap_{(r, r+1)}]}{\rmP[\swap_{(r, r+1)}]}.
\end{align*}
This justifies using DEO rather than SEO.
Importantly for us, $2 \tau$ is the fraction of samples, starting from $\pi_R$, that make their way down to the target $\pi_1$.

Note that ELE will be violated if chains are stuck in different modes, and the modes do not have identical energy surfaces. In other words, we expect ELE to be violated in most multi-modal problems. Nonetheless, \cite{Syed2022-ag} finds that the results of this section roughly held in a variety of problems despite ELE being violated.

As $\max_r|T_r^{-1}-T_{r+1}^{-1}|\to0$, the swap probabilities are governed by the increasing function, $\Lambda(T)$, which satisfies
\begin{theorem}[\cite{Syed2022-ag} Theorem 2]
  \label{theorem:syed-theorem-2}
  For annealing schedule $1=T_1<T_2<\cdots<T_R\leq\infty$,
  \begin{align*}
    1 - \swapprob &= \Lambda(T_{r+1}) - \Lambda(T_r) + O(\max_r |T_r^{-1} - T_{r+1}^{-1}|^3).
  \end{align*}
\end{theorem}
When $T_R=\Tmax$ is fixed, $\Lambda$ is understood to mean $\Lambda(\Tmax)$, the \emph{global communication barrier}.

Ignoring the error term in theorem \ref{theorem:syed-theorem-2}, the round trip rate $\tau_{DEO}$ is optimized when
\begin{align}
  \label{align:optimal-swap-rate}
1 - \swapprob\equiv \Lambda/R.
\end{align}
Consider running $k$ copies (chains, in our terminology) of REMC independently, with a total computational budget of $\bar R$. In other words, $\bar R = k\,R$. In this setup, \cite{Syed2022-ag} derives the optimal number of chains $k^\ast$, number of replicas $R^\ast$, and round trip rate $\tau_{DEO}^\ast$.
\begin{align}
  \label{align:optimal-num-replicas}
  \begin{split}
    R^\ast &= 2\Lambda + 1,
    \quad k^\ast = \frac{\bar{R}}{R^\ast} = \frac{\bar R}{2\Lambda + 1},
    \quad \tau_{DEO}^\ast = \frac{k^\ast}{2 + 4\Lambda} = \frac{\bar R}{2(2\Lambda+1)^2}.
  \end{split}
\end{align}
This optimum is achieved when $\swapprob\equiv\Lambda/R\approx0.5$, although they recommend $\swapprob>0.5$ to reduce the ELE violation.

An asymptotic expression for $\Lambda(T_R)$ is also provided in the~\iid~regime.
Here, one adds a parameter $d$, and with $N=d\cdot N'$, assumes the prior and likelihood act in an~\iid~manner on each of the $d$ copies of $\Rone^{N'}$. In other words,
\begin{align}
  \label{align:iid-decomposition}
      \logposterior &=\sum_{i=1}^d V(x_i),\quad \logprior = \sum_{i=1}^d \tilde V_0(x_i).
\end{align}
We re-state their proposition, extending it to posterior tempering.
\begin{proposition}[\cite{Syed2022-ag} proposition 4]
  \label{proposition:syed-proposition-4}
  Given \ref{align:iid-decomposition}, as $d\to\infty$, we have asymptotic convergence
  \begin{align*}
    \Lambda_d(T_R) &\stackrel{asy.}{\sim}
    \sqrt{\frac{1}{\pi}}\int_1^{T_R} \frac{\sigma(T)}{T^2}\d T,
  \end{align*}
  where, with $X(T)$ the tempered state,
  \begin{align*}
    \sigma^2(T) &= \left\{ 
    \begin{matrix}
      \Var{\log p(X(T)\g y)},&\quad \mbox{posterior tempering}\\
      \Var{\log p(y\g X(T))},&\quad \mbox{likelihood tempering}.
    \end{matrix}
    \right.
  \end{align*}
\end{proposition}

In this~\iid~regime, $\sigma^2 = O(d)$ and hence $\Lambda = O(\sqrt{d})$. It follows that $\bar R / \tau_{DEO}^\ast = O(d)$ units of work are done to produce each sample making its way from the hottest replica to the target.

\subsubsection{Number of replicas and its relation to under-constrained problems}
\label{subsubsection:number-of-replicas-scaling}
This section makes its point by example, although results should apply more generally.
The example is a linear/Gaussian problem, with $A\in\RMN$ and $\mbox{Rank}(A) = M < N$:
\begin{align*}
    \prior &\propto \exp\left\{ -\frac{1}{2}x^T C_{pr}^{-1}x \right\},
    \quad \likelihood \propto \exp\left\{ -\frac{1}{2\sigma^2}\|Ax - y\|^2 \right\}.
\end{align*}
In this case, the log posterior is a sum of $N$ terms, whereas the log likelihood a sum of $M$. If these terms are independent enough, we expect proposition \ref{proposition:syed-proposition-4} and \eqref{align:optimal-num-replicas} to show the optimal number of replicas, $R^\ast$, is $= O(\sqrt{N})$ for posterior tempering, and $= O(\sqrt{M})$ for likelihood tempering. Figure \ref{fig:tempering-comparison} shows this relation holds for a toy problem.

More precise results can be obtained for unimodal distributions. This approach is partially justified by noting that if $R$ replicas are needed to exchange in unimodal density $\pi$, at least $R$ should be needed to exchange in a multimodal distribution where one mode looks like $\pi$.

Consider an arbitrary Gaussian posterior, and posterior tempering. Proposition \ref{proposition:swap-probs-as-typical-set-overlap} shows that, with $\chi^2_N$, $\tilde\chi^2_N$ two independent chi-square random variables,
\begin{align*}
  \swapprob = \rmP\left[ T_r\chi^2_N - T_{r+1}\tilde\chi^2_N > 0 \right].
\end{align*}
Since the mean of $T_r\chi^2_N - T_{r+1}\tilde\chi^2_N$ is $N(T_r - T_{r+1})< 0$, the probability is non-vanishing as $N\to\infty$ only if the standard deviation is of the same order. This is satisfied if $T_{r+1}/T_r= 1 + c/\sqrt{N}$, for some $c>0$ depending only on the desired acceptance probability. If $\Tmax<\infty$ is chosen ahead of time, $\Tmax = (T_2/T_1)^{R-1} = (1 + c/\sqrt{N})^{R-1}$. It follows that $R\propto\sqrt{N}\log\Tmax$. This line of reasoning can be extended to any distribution with constant heat capacity \cite{Kofke2002-my}.

To analyze the case of likelihood tempering, it will help to re-write the variance of the potentials from proposition \ref{proposition:syed-proposition-4} as specific heat like quantities:
\begin{lemma}
  \label{lemma:specific-heat-relation}
  Let $X(T)\sim\pi^{1/T}$ be a tempered state. Then,
  for posterior tempering,
  \begin{align*}
    \Var{\log p(X(T)\g y)} &= -T^2\frac{\d}{\dT}\Exp{\log p(X(T)\g y)},
  \end{align*}
  and for likelihood tempering,
  \begin{align*}
    \Var{\log p(y\g X(T))} &= -T^2\frac{\d}{\dT}\Exp{\log p(y \g X(T))}.
  \end{align*}
\end{lemma}
With this in hand, suppose the prior covariance $C_{pr}=I$ in \eqref{align:gaussian-inverse-problem}, and $\mbox{Rank}(A)=M\leq N$. Suppose $A$ has singular values $\{\alpha_n\}$. If the nonzero $\alpha_n$ are all equal, one can use lemma \ref{lemma:specific-heat-relation} and proposition \ref{proposition:syed-proposition-4} to derive an asymptotic relation (as $M\to\infty$) for likelihood tempering. Assuming the relation holds when $\alpha_n$ is non-constant, we have
\begin{align*}
  \Lambda &\stackrel{asy.}{\sim}
  \sqrt{\frac{1}{2\pi}} \int_0^1 \sqrt{\sum_{n=1}^N \frac{\alpha_n^4}{\left( \beta\alpha_n^2 + \sigma^2 \right)^2}}\d\beta
  \leq \sqrt{\frac{M}{2\pi}} \log \left[ \frac{1 + (\sigma/\alpha_1)^2}{(\sigma/\alpha_1)^2} \right].
\end{align*}
\eqref{align:optimal-num-replicas} now indicates that the optimal number of replicas should be $= O(\sqrt{M})$.

\subsection{Selecting parameters for REMC with HMC}
\label{subsection:hmc-params-with-remc}
The work of previous sections allow us to give concrete recommendations.

\subsubsection{Likelihood vs.~posterior tempering}
We usually prefer likelihood tempering. The reasons are that it (i) has better scaling properties for under-constrained problems (section \ref{subsubsection:number-of-replicas-scaling}), (ii) is stable no matter how large $\Tmax$ is, (iii) the hottest replica is close to the prior, facilitating the leapfrog heuristic of section \ref{subsubsection:selecting-number-of-leapfrog-steps-for-all-replicas}.
\begin{figure}[h]
  \begin{center}
    \includegraphics[width=0.9\textwidth]{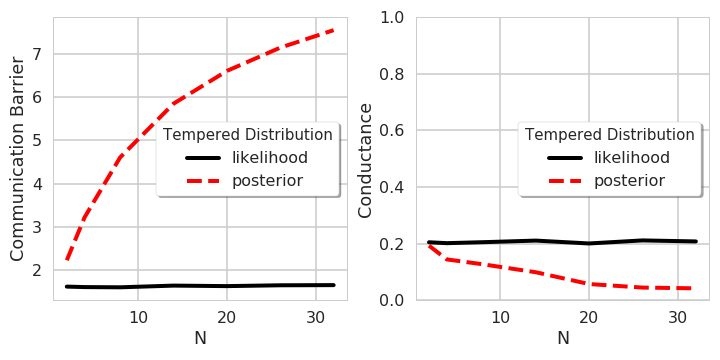}
  \end{center}
  \caption{
  {\bf Tempering comparison: } Communication barrier, $\Lambda$, and conductance in the bi-modal normal of \eqref{align:multi-model-toy-model}.  $N$ is increased while fixing $M=1$. These show the $\sqrt{N}$, $\sqrt{M}$ scaling of the communication barrier for posterior tempering discussed in section \ref{subsubsection:number-of-replicas-scaling}. The resultant conductance is much better for likelihood tempering.
  }
  \label{fig:tempering-comparison}
\end{figure}

Likelihood tempering is limited in that the highest temperature replica, is approximately the prior.
For that reason, if there is a possibility of modes far outside the typical set of the prior, posterior tempering should be used.

\subsubsection{Selecting the annealing schedule}
\label{subsubsection:selecting-replica-temperatures}
The annealing schedule is $1=T_1<T_2<\cdots<T_R=\Tmax$.  We select this in two stages. First we find $\Tmax$, and second we adjust the other temperatures.

$\Tmax$ must be hot enough so that the hottest replica can mix well and explore the proximity of all posterior modes.  The dimension $N$ need not play a role. For example, with our spectroscopy model, tempering changes the noise term $\sigma(x)^2$ (in \eqref{align:spectroscopic-likelihood}) to $T\sigma(x)^2$.  We therefore expect $\Tmax\propto \sigma^{-2}$ to be hot enough for non-vanishing conductance (see section \ref{subsection:conductance-basics}). The constant of proportionality will depend on the maximum residual $|y - \left[ \calI\varphi_\nu(x) \right]_m|$ on some path between modes.  We cannot expect to know this a-priori. Instead, during burn-in, we start with geometrically increasing temperatures. We then monitor $\Rhat$ (\cite{Vehtari2021-ai}) to see which replicas mix \emph{without} swapping. $\Tmax$ is set to the coldest temperature that was mixing during this burn-in phase.
For likelihood tempering, we \emph{could have} set $\Tmax=\infty$, but we found that to be less efficient.

After $T_R=\Tmax$ is chosen, we adjust $T_2,\ldots,T_{R-1}$ until $\swapprob$ is close to constant, as justified by \eqref{align:optimal-swap-rate}. We use an interpolation scheme as in \cite{Syed2022-ag}.

\subsubsection{Selecting step sizes for each replica}
\label{subsubsection:selecting-replica-step-sizes}
A good initial guess for the $r^{th}$ replica's step size, $h_r$, is $h_r\propto \sqrt{T_r}$, since this would be ideal for a Gaussian.  After every temperature change, step sizes can be adjusted by an iterative scheme \cite{Andrieu2008-xv}. It helps to have an initial guess for the new step sizes. This can be done by finding an interpolating function, $h = f(T)$, that is piecewise linear in $\sqrt{T}$.

\subsubsection{Selecting the number of leapfrog steps}
\label{subsubsection:selecting-number-of-leapfrog-steps-for-all-replicas}
To take advantage of batch operations on a single GPU, the number of leapfrog steps $\ell$ should be the same for all replicas (see section \ref{subsection:computational-efficiency}). Given likelihood tempering, the hottest replica will be similar to the prior, and should be much better conditioned than the target. We therefore expect the hottest replica will need far fewer leapfrog steps than the target.
Since we choose $\Tmax$ so that only the hottest replica is mixing well (without swapping), we should choose $\ell$ to facilitate mixing of this hottest replica.  This allows $\ell\propto1/h_R$, which is less than the $O(1 / h_1)$ needed for the target to mix well (without swapping).

As it turns out, we can reduce $\ell$ even further. Equation \eqref{align:round-trip-rates} indicates that only one out of $1 + \gamma$ samples produced by the hottest replica make their way to the target. These $1+\gamma$ intermediate samples give additional time for the hottest replica to mix. To use this opportunity, note that leapfrog integration travels for a time $\approx h_R\ell$, and if $h_R\ell$ is small, sampling approximates a random walk. The expected time traveled by these $1 + \gamma$ intermediate samples is therefore proportional to $h_R\ell\,\sqrt{1 + \gamma}$.  Assuming the hottest replica has largest scale of $\lambda_{R,1}$, we should choose $\ell$ so that $h_R\ell\sqrt{1 + \gamma} \approx \lambda_{R,1}\pi/2$, or
\begin{align}
  \label{align:leapfrog-heuristic}
  \ell  &\approx  \lambda_{R,1}\,\frac{\pi}{2}\,\frac{1}{h_R\,\sqrt{1+\gamma}}.
\end{align}

Reducing $\ell$ according to this \emph{leapfrog heuristic} means the integration length is less than ideal for every replica. As a result, state samples $X^j$ will certainly not be independent. We also expect the potential energy samples, $-\log[p(y\g X^j)]$ to lack independence as well. In other words, the ELE assumption of \cite{Syed2022-ag}, (see section \ref{subsubsection:existing-fundamental-results}), will be violated.
This can be visualized by plotting chain traces (see figure \ref{fig:U-traces-toy-model}). Swapping is turned off for these plots, since we want to visualize \emph{local} exploration.
\begin{figure}[h]
  \begin{center}
    \includegraphics[width=\textwidth]{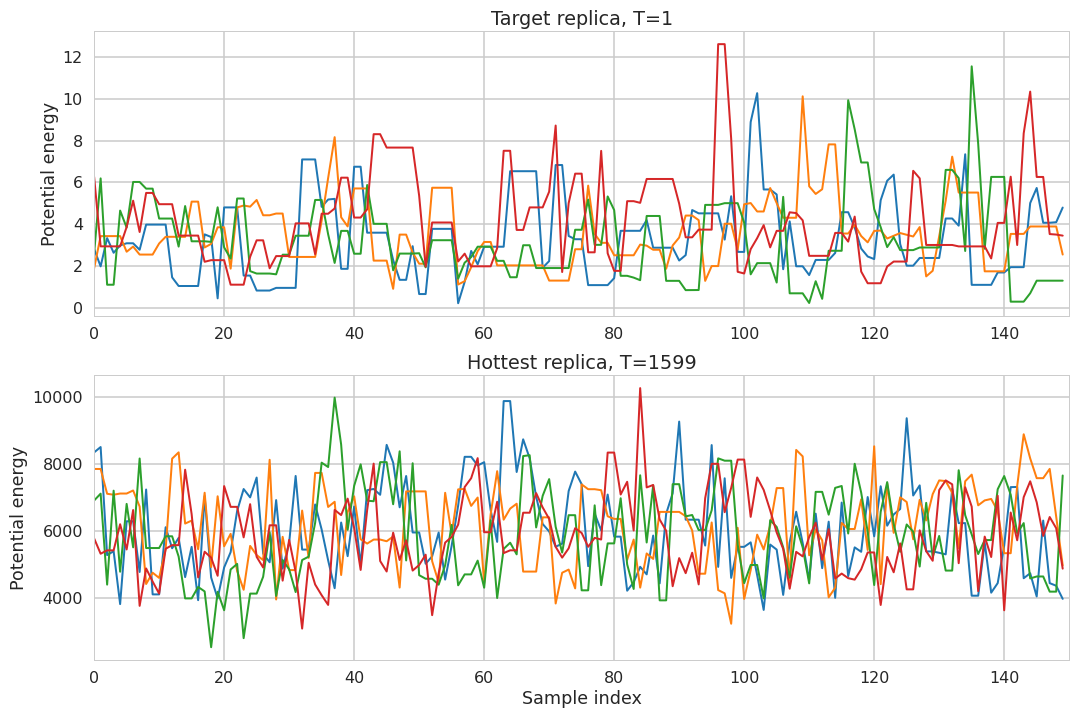}
  \end{center}
  \caption{
  {\bf Potential energy traces }from the toy model \eqref{align:multi-model-toy-model}. 
  Three chain traces of $-\log[p(y\g X)]$ are shown for the target $T=1$ and hottest $T=1599$ replicas. Autocorrelation length is longer in the target. In neither case are samples independent. Since the two modes in the toy model are identical, one cannot tell (by potential energy alone) if the chains are in different modes.
  }
  \label{fig:U-traces-toy-model}
\end{figure}

The toy model in figure \ref{fig:U-traces-toy-model} implemented \eqref{align:multi-model-toy-model} with $N=10$, $M=5$, 55 replicas, and noise $\sigma=0.025$.
We also set $\Tmax=\sigma^{-2}$ to ensure the hottest replica was barely mixing. Other temperatures were geometrically distributed. We estimated $\lambda_{R,1}$ using an exact formula for the largest scale of the tempered unimodal version of \eqref{align:multi-model-toy-model}.

We next used the same toy model, but with swapping turned on ($\swapprob\approx0.75$). Noise level $\sigma$ was swept from 0.00075 to 0.05 while keeping the number of replicas $R$ fixed at 55. This varies the schedule inefficiency $\gamma$.
See figure \ref{fig:leapfrog-multiplier-test-toy-model}, where following our leapfrog heuristic exactly (\emph{leapfrog multiplier} = 1) led to more efficient sampling than shorter (\emph{leapfrog multiplier} $<1$) or longer (\emph{leapfrog multiplier} $>1$) integration times.
\begin{figure}[h]
  \begin{center}
    \includegraphics[width=\textwidth]{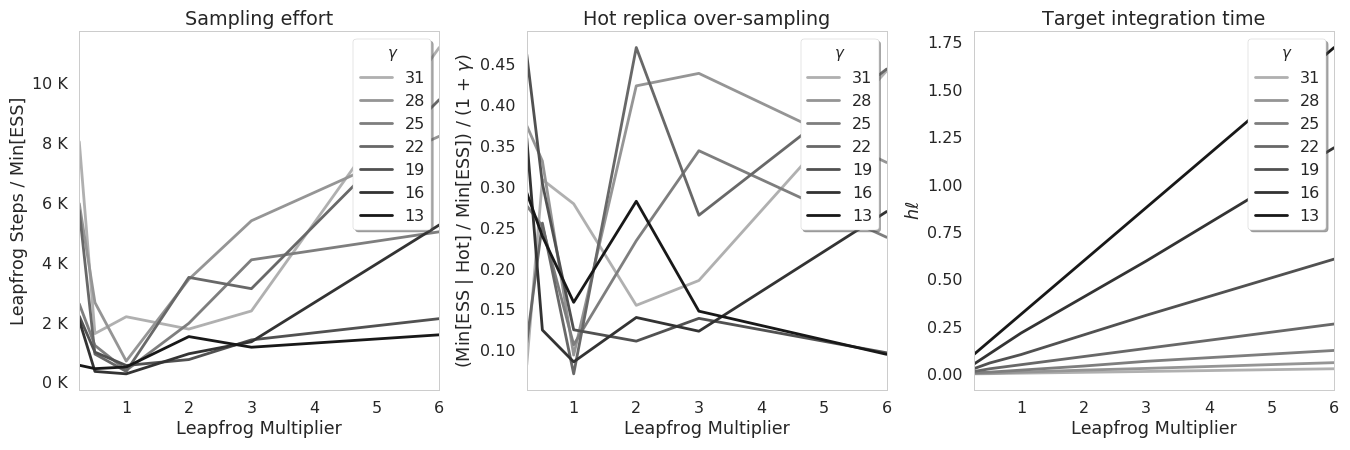}
  \end{center}
  \caption{
  {\bf Optimal number of leapfrog steps and $\gamma$, toy model:} The (likelihood) tempered posterior from \eqref{align:multi-model-toy-model} was run for a number of different $\sigma$.  Each $\sigma$ corresponds to a different schedule inefficiency $\gamma$. The x-axis is a \emph{leapfrog multiplier}.  The number of leapfrog steps taken is this multiplier times the heuristic \eqref{align:leapfrog-heuristic}.
  {\bf Left:} Sampling effort, as measured in leapfrog steps per $\miness$, is lowest when the leapfrog multiplier $=1$. So following the heuristic \eqref{align:leapfrog-heuristic} exactly was most efficient.
  {\bf Center:} The number of effective samples obtained by the hottest replica (in relation to the target, and $\gamma$) was 1/2 to 1/10 that predicted by \eqref{align:round-trip-rates}.  Perhaps extra effective samples were picked up along the ``hot to target'' trip.
  {\bf Right:} The target integration time, $h_1\ell$, is significantly less than would be required to efficiently explore the largest target scale (=1) without REMC. So (i) the leapfrog heuristic is saving a significant number of steps, (ii) the state samples are correlated, hinting that ELE is likely violated, (iii) even so, following our leapfrog heuristic was the most efficient choice.
  }
  \label{fig:leapfrog-multiplier-test-toy-model}
\end{figure}

In a problem where modes are not Gaussian, the optimal integration time will no longer be given by the $\pi/2$ heuristic, and the energy distribution may be such that an ELE violation is more problematic.
To test these effects, we used the multi-modal version of our spectroscopy model from section \ref{subsection:multi-modal-posteriors}.
Since our REMC implementation runs $R$ concurrent chains on GPU at once, memory usage increases $R$ times. We therefore limited $R$ to 150.
We used the parameter choices outlined in sections \ref{subsubsection:selecting-replica-temperatures}, \ref{subsubsection:selecting-replica-step-sizes} and assumed $\lambda_{R,1}=1$.
Figure \ref{fig:U-traces-real-model} shows potential energy traces, similar to the toy-model traces of figure \ref{fig:U-traces-toy-model}. Again we see correlated samples. However, with the spectroscopy model, different modes have different average potential energy for $T=1$.
\begin{figure}[h]
  \begin{center}
    \includegraphics[width=\textwidth]{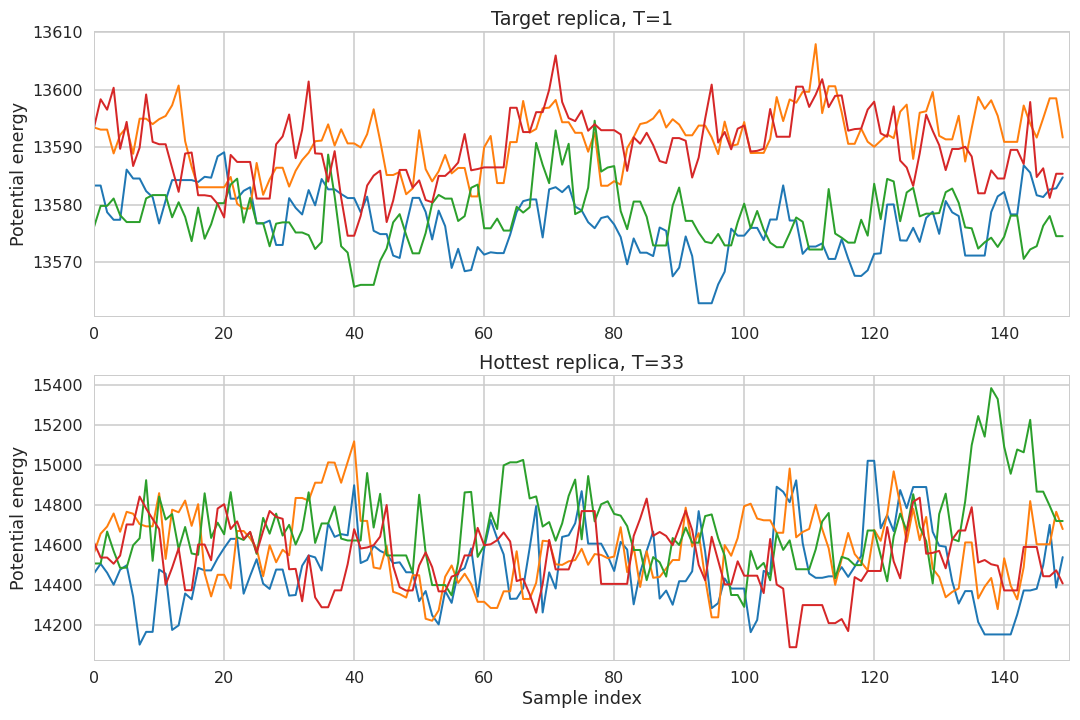}
  \end{center}
  \caption{
  {\bf Potential energy traces }from the spectroscopy model.
  Three chain traces shown for the target $T=1$ and hottest $T=33$ replicas. Both replicas exhibit autocorrelation within each chain. The target replica's chains also exhibit different average potential energy. This happens since the target replica is unable to jump between modes, and the modes have different potential energy.
The autocorrelation of each chain (not shown) is longer for the target $T=1$ than hottest replica $T=33$. This is is expected since both use the same number of leapfrog steps, but the target must use a smaller step size.
  }
  \label{fig:U-traces-real-model}
\end{figure}

As with the toy model, we used a \emph{leapfrog multiplier} to test variations on the leapfrog heuristic \eqref{align:leapfrog-heuristic}. Sampling is once again most efficient if the heuristic is followed exactly (Fig.~\ref{fig:leapfrog-multiplier-test-real-model}). We also compared setting $\Tmax=\infty$ in addition to the coldest temperature that mixes without swapping. $\Tmax=\infty$ was much less efficient, supporting the recommendations given in section \ref{subsubsection:selecting-replica-temperatures}.
\begin{figure}[h]
  \begin{center}
    \includegraphics[width=\textwidth]{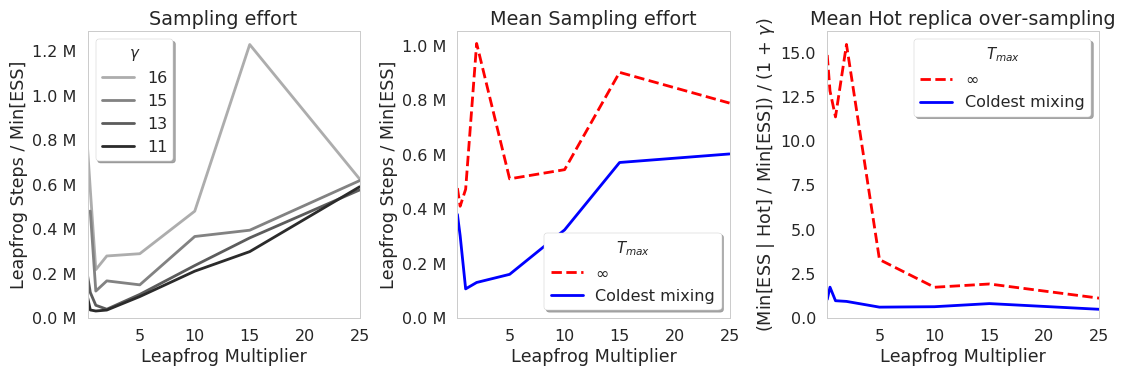}
  \end{center}
  \caption{
  {\bf Efficiency, number of leapfrog steps, and $\gamma$, in the spectroscopy model:} The (likelihood) tempered multi-modal posterior from section \ref{subsection:multi-modal-posteriors}. X and Y axis are as in figure \ref{fig:leapfrog-multiplier-test-toy-model}. The ``Mean'' is taken over experiments with different $\gamma$.
  {\bf Left:} Following the leapfrog heuristic was the most efficient choice.
  {\bf Center:}
  Experiments were also done with $\Tmax=\infty$ in addition to the coldest temperature that mixes without swapping (recommended in section \ref{subsubsection:selecting-replica-temperatures}).
  When $\Tmax=\infty$, sampling was less efficient, and the leapfrog heuristic selected too few steps to be efficient.
  {\bf Right:} When $\Tmax=\infty$, the ratio of (effective) samples drawn by the hottest replica to that of the target was greater than $1+\gamma$. This is unexpected and unexplained. For both $\Tmax$ values, once $\ell$ is high enough, \eqref{align:round-trip-rates} is satisfied (the $Y$ value is $\approx1$), presumably because ELE is no longer violated.
  }
  \label{fig:leapfrog-multiplier-test-real-model}
\end{figure}

\section*{Acknowledgements}
The author's would like to acknowledge the entire TensorFlow probability team at Google. In particular, Colin Carroll's feedback was helpful in preparation. Also, Saifuddin Syed was helpful in confirming the correctness of our re-framing of results from \cite{Syed2022-ag}.

\appendix

\section{Supporting Lemmas}
\label{appendix:lemmas}

\subsection{Equilbration of matrices}
\label{appendix:subsection:equilibration}
The equilibration of a matrix $A$ is defined in \cite{Van_der_Sluis1969-vh} as a rescaling of the rows $A_j$ such that $\|A_j\|\equiv1$, in some vector norm.  Since the preconditioned covariance $(D^{-1}L)(D^{-1}L)^T$ has ones on the diagonal, $D^{-1}L$ is equilibrated in the L2 norm. They go on to prove that equilibration is a near optimal diagonal scaling. In this section we follow their lead, re-proving pieces from scratch in order to avoid unwinding their more general, and technical, results.

Below, $\|x\|_2$ is the L2 norm on vectors,$\|A\|_2 := \max_{\|x\|_2=1}x\cdot Ax$ is the spectral norm on matrices. This is upper bounded by the Frobenius norm $\|A\|_F^2 = \sum_{i,j}A_{ij}^2$.
In the following,  $A_j$ denotes the $j^{th}$ row of matrix $A$.

A key to \cite{Van_der_Sluis1969-vh} establishing their results is the notion of monotonicity of some norms. We prove here the only case we need.
\begin{lemma}
  \label{lemma:schatten-norm-monotonicity}
  Let $G\in\RNN$ be diagonal with non-negative entries.
  For any $A\in\RNN$,
  \begin{align*}
    \|AG\|_{S^4} \geq \|A\|_{S^4} \min_n\left\{ G_{nn} \right\}.
  \end{align*}
\end{lemma}
\begin{proof}
  We have
  \begin{align*}
    \|AG\|_{S^4}^4
    &= \Trace{ (AG)(AG^T)(AG)(AG)^T}
    = \Trace{ G^2 A^TA G^2 A^TA}.
  \end{align*}
  The result then follows by repeatedly applying (to positive semi-definite $U$)
  \begin{align*}
    \Trace{ GU }
    &= \sum_n G_{nn} U_{nn}
    \geq \min_n \left\{ G_{nn} \right\} \Trace{U}.
  \end{align*}
\end{proof}

\begin{lemma}
  \label{lemma:l2-for-k-nonzeros}
  Let $A\in\RNN$ be equilibrated in the L2 norm.  Suppose $AA^T$ has $K$ nonzero elements in every row/column. Then $\|A\|_2\leq \sqrt{K}$.
\end{lemma}
\begin{proof}
  Since $(AA^T)_{ij} = A_i\cdot A_j$, our condition on $AA^T$ implies that, for every $i$, there are at most $K$ rows $A_j$ such that $A_i\cdot A_j \neq 0$. For these rows, Cauchy-Schwarz gives us $A_i\cdot A_j\leq 1$.  Therefore,
  \begin{align*}
    \|A\|_2^4 = \|AA^T\|_2^2 \leq \|AA^T\|_F^2 = \sum_{i,j} (A_i\cdot A_j)^2
    \leq \sum_{\left\{ (i, j) \st A_i\cdot A_j\neq0 \right\}} 1
    \leq K^2.
  \end{align*}
  The result follows by taking fourth roots.
\end{proof}

We now have our main equilibration result
\begin{lemma}
  \label{lemma:equilibration-almost-optimal}
  Suppose $A$ is equilibrated in the L2 norm, with $AA^T$ having at most $K$ nonzero elements in every row. Then for any diagonal matrix $G$,
  \begin{align*}
    \W(A)
    = \|A\|_2 \|A^{-1}\|_{S^4}
    \leq \sqrt{K}\,\|G^{-1}A\|_2 \|A^{-1}G\|_{S^4}
    = \sqrt{K}\,\W(G^{-1}A).
  \end{align*}
\end{lemma}
\begin{proof}
  For any matrix $F$, $\|F\|_2 \geq \max_j \|F_j\|_2$. Therefore, using lemma \ref{lemma:l2-for-k-nonzeros},
  \begin{align*}
    \|G^{-1}A\|_2
    \geq \max_j \|(G^{-1}A)_j\|_2
    = \max_j |G_{jj}^{-1}|
    = \max_j |G_{jj}^{-1}| \frac{\|A\|_2}{\sqrt{K}}.
  \end{align*}
  Likewise, lemma \ref{lemma:schatten-norm-monotonicity} tells us $\|A^{-1}G\|_{S^4}\geq \|A^{-1}\|_{S^4}\min_j |G_{jj}|$. Therefore,
  \begin{align*}
    \|G^{-1}A\|_2 \|A^{-1}G\|_{S^4}
    \geq\max_j |G_{jj}^{-1}| \frac{\|A\|_2}{\sqrt{K}} \|A\|_{S^4}\min_j |G_{jj}|
    =\frac{\|A\|_2}{\sqrt{K}} \|A\|_{S^4}.
  \end{align*}
  Rearranging, we have proved lemma \ref{lemma:equilibration-almost-optimal}.
\end{proof}

\subsection{Chi-square bounds}
\begin{lemma}
  \label{lemma:chi-square-bound}
  If $Z_s\sim\calN(0, 1)$ are~\iid, then for $\eps\in(0, 1)$,
  \begin{align*}
    \rmP
    \left[
    \left|
    \frac{1}{S}\sum_{s=1}^S Z_s^2 - 1 - \left( \frac{1}{S}\sum_{s=1}^S Z_s\right)^2
    \right|
    \geq \eps
    \right]
    &\leq 3 \exp\left\{ \frac{-S \eps^2}{25} \right\}.
  \end{align*}
\end{lemma}
\begin{proof}
  In the corollary to lemma 1 in \cite{Laurent2000-xw}, they establish, for $\chi^2_S$ a chi-square random variable with $S$ degrees of freedom, and $c>0$,
  \begin{align}
    \label{align:chi-square-tail-bound}
    \begin{split}
      \rmP[\chi^2_S - S\geq 2 \sqrt{St} + 2 t] & \leq e^{-t},\\
      \rmP[\chi^2_S - S\leq -2 \sqrt{St}] & \leq e^{-t}.
    \end{split}
  \end{align}
  The second inequality in \eqref{align:chi-square-tail-bound} directly gives us, if $\eps/2 = 2 \sqrt{c/S}$,
  \begin{align*}
    \rmP\left[ \frac{1}{S}\sum_{s=1}Z_s^2 - 1 \leq -\frac{\eps}{2} \right]
    &= \rmP\left[ \chi_S^2 - S \leq -2\sqrt{Sc} \right]
    \leq e^{-c}
    = e^{-S\eps^2/16}.
  \end{align*}
  To use the first inequality in \eqref{align:chi-square-tail-bound}, we start with the substitution $\eps/2 = (5/2)\sqrt{c/S}$, which implies $\sqrt{c/S} = \eps/5 < 1/4$, so that $(5/4) \geq 1 + \sqrt{c/S}$:
  \begin{align*}
    \rmP \left[ \frac{1}{S}\sum_{s=1}^S Z_s^2 - 1 \geq \frac{\eps}{2} \right]
    &= \rmP \left[ \frac{\chi^2_S - S}{S} \geq 2\frac{5}{4}\sqrt{\frac{c}{S}} \right] \\
    &\leq \rmP \left[ \frac{\chi^2_S - S}{S} \geq 2\left( 1 + \sqrt{\frac{c}{S}} \right)\sqrt{\frac{c}{S}} \right] \\
    &= \rmP \left[ \chi^2_S - S \geq 2\sqrt{Sc} + 2c \right] \\
    &\leq e^{-c} \\
    &= e^{-S\eps^2/25}.
  \end{align*}
  Combining, we have
  \begin{align}
    \label{align:centered-chi-square-tail-bound}
    \rmP\left[
      \left|
        \frac{1}{S}\sum_{s=1}^S Z_s^2 - 1
      \right|
      \geq \frac{\eps}{2}
    \right]
    \leq 2 e^{-S\eps^2/25}.
  \end{align}

  Next, a Bernstein inequality gives the standard Normal tail bound
  \begin{align}
    \label{align:normal-related-tail-bound}
    \rmP\left[ 
      \left( \frac{1}{S}\sum_{s=1}^S Z_s \right)^2 \geq \frac{\eps}{2}
    \right]
    \leq e^{-S\eps/4}.
  \end{align}
  Combining \eqref{align:centered-chi-square-tail-bound} and \eqref{align:normal-related-tail-bound}, we have the result.
\end{proof}

\medskip

\bibliographystyle{unsrt}



\end{document}